\documentclass[3p,times]{elsarticle}

\newcommand{\ON}{\mathsf{ON}}
\newcommand{\OFF}{\mathsf{OFF}}
\newcommand{\No}{\mathsf{N}}
\newcommand{\Ea}{\mathsf{E}}
\newcommand{\We}{\mathsf{W}}
\newcommand{\So}{\mathsf{S}}
\newcommand{\rread}{\mathit{read}}
\newcommand{\charge}{\mathit{charge}}
\newcommand{\discharge}{\mathit{discharge}}
\newcommand{\bat}{\mathit{bat}}
\newcommand{\field}{\mathit{field}}
\newcommand{\battery}{\mathit{battery}}
\newcommand{\robot}{\mathit{robot}}
\newcommand{\move}{\mathit{move}}
\newcommand{\loc}{\mathit{loc}}

\newcommand{\energy}{\mathit{energy}}
\newcommand{\nooverlap}{\mathit{no-overlap}}
    
\newcommand{\pr}{\mathsf{pr}}
\newcommand{\Rp}{\mathbb{R}_+}

\newcommand{\TES}[2]{\mathit{TES}(#1)}

\newcommand{\E}{\mathbb{E}}
\newcommand{\N}{\mathbb{N}}
\newcommand{\overbar}[1]{\mkern 1.5mu\overline{\mkern-1.5mu#1\mkern-1.5mu}\mkern 1.5mu}

\DeclareRobustCommand{\mmodels}{\mathrel{|\mkern-2mu|}\joinrel\Relbar}
\newcommand{\Po}{\mathcal{P}}
\newcommand{\tes}{TES}
\newcommand{\red}[1]{\short{#1}}

\newcommand{\short}[1]{\ifshort #1 \fi}
\newif\ifshort
\shortfalse
\usepackage{comment}

\usepackage{siunitx}

\usepackage{amssymb}
\usepackage{amsthm}

\newtheorem{notation}{Notation}
\newtheorem{definition}{Definition}
\newtheorem{example}{Example}
\newtheorem{lemma}{Lemma}

\newtheorem{property}{Property}
\newtheorem{corollary}{Corollary}

\usepackage{xcolor}
\usepackage{stmaryrd}
\SetSymbolFont{stmry}{bold}{U}{stmry}{m}{n}

\usepackage{makecell}

\usepackage{mathtools}

\begin{document}
\begin{frontmatter}

\title{A Semantic Model for Interacting Cyber-Physical Systems}

\author[LEIDEN]{Benjamin Lion }
\ead{b.lion@liacs.leidenuniv.nl}
\author[CWI, LEIDEN]{Farhad Arbab }
\ead{farhad@cwi.nl}
\author[SRI]{Carolyn Talcott }
\ead{carolyn.talcott@gmail.com}
\address[LEIDEN]{Leiden University, Leiden, The Netherlands}
\address[CWI]{CWI, Amsterdam, The Netherlands}
\address[SRI]{SRI International, CA, USA}

\begin{abstract}
We propose a component-based semantic model for Cyber-Physical Systems (CPSs) wherein the notion of a component abstracts the internal details of both cyber
and physical processes, to expose a uniform semantic model of their externally observable behaviors expressed as sets of sequences of observations.
We introduce algebraic operations on such sequences to model different kinds of component composition.
These composition operators yield the externally observable behavior of their resulting composite components through specifications of interactions of the behaviors of their constituent components, as they, e.g., synchronize with or mutually exclude each other's alternative behaviors.  Our framework is expressive enough to allow articulation of properties that coordinate desired interactions among composed components within the framework, also as component behavior. We demonstrate the usefulness of our formalism through examples of coordination properties in a CPS consisting of two robots interacting through shared physical resources.
\end{abstract}

\begin{keyword}
    cyber-physical \sep system
\end{keyword}

\end{frontmatter}

\section{Introduction}

Compositional approaches in software engineering reduce the complexity of
specification, analysis, verification, and construction of software by decomposing a large system into (a) smaller parts, and (b) their interactions. Applied recursively, compositional methods reduce software complexity by breaking the software and its parts into ultimately simple modules, each with a description, properties, and interactions of manageable size. The natural tendency to regard each physical entity as a separate module in a Cyber-Physical System (CPS) makes compositional methods particularly appealing for specification, analysis, verification, and construction of CPSs. However, the distinction between discrete versus continuous transformations in modules representing cyber versus physical processes complicates the semantics of their specification and their treatment by requiring: (1) distinct formalisms to model discrete and continuous phenomena; (2) distinct formalisms to express compositions and interactions of cyber-cyber, cyber-physical, and physical-physical pairs of modules; and (3) when to use which formalism to express composition and interactions of hybrid cyber-physical modules.
Our work is distinguished from existing work in the following three senses.

First, we unify cyber and physical aspects within the same semantic model. While this feature is present in some other works (e.g., signal semantics for cyber-physical systems \cite{LML06,TSSL13}, time data streams for connectors \cite{AR03}), we add a structural constraint by defining an observable as a set of events that happen at the same time. An observation is therefore a pair $(O,t)$ of an observable $O$ occurring at a time $t$. As a result, we abstract the underlying data flow (e.g., causality rules, input-output ports) that must be implemented for such observables to happen in other models. We believe that this abstraction is different from traditional approaches to design cyber-physical systems and, for instance, may naturally compose a set of observations occurring at the same time into a new observation formed by the union of their observables. The Timed-Event Stream (TES) model proposed in this paper differs from the trace semantics in \cite{AR03} in that it explicitly and directly expresses synchronous occurrences of events. Like the one in \cite{AR03} but unlike many other trace semantics that effectively assume a discrete model of time, the TES model is based on a dense model of time. These distinctions become significant in enabling a compositional semantic model where the sequences of actions of individual components/agents are specified locally, not necessarily in lock-step with those of other entities. It is on this basis that we can define our expressive generic composition operators with interesting algebraic properties. The advantages of this compositional semantics include not just modular, reusable specification of components, but also modular abstractions that allow reasoning and verification following the assume/guarantee methodology.

Second, we make coordination mechanisms explicit and exogenous to components. 
A component is a standalone entity that exhibits a behavior (which may be described by a finite state automaton, hybrid automaton, or a set of differential equations), and permissible coordinated interactions among components is given by a set of constraints on behaviors of each component. One such composition operation, also widely used in the design of modular systems~\cite{AR03} is set intersection: each component interacts with other components by producing a behavior that is consistent with their shared events. We generalize such composition operations to ease the specification of permissible interactions among cyber-physical components. We also show the benefit in modeling interaction exogenously when it comes to reasoning about, e.g., asymmetric product operations, or proving some algebraic properties. 
    Our approach, in this sense, differs from how control theory models interacting cyber-physical systems, where differences between cyber and physical are first class and lead to, for instance, the use of hybrid models~\cite{DBLP:journals/jsc/LafferrierePY01}.
    Instead, we assume that the underlying control is given, and consider the coordination in a discrete framework~\cite{L19}, where (cyber or physical) processes are modeled as components that exhibit sequences of timed-events and relations on those components act as constraints.

Third, we give an alternative view on satisfaction of trace and behavioral properties of cyber-physical systems. We expose properties as components and show how to express coordination as a satisfaction problem (i.e., adding to a system of components a coordinator that restricts each component to a subset of its behavior to comply with a trace property). 
Moreover, we show that trace properties are not adequate to capture all important properties.
We introduce behavioral properties, which are analogous to hyperproperties of~\cite{CS10}.
We show, for instance, how the energy adequacy property in a cyber-physical system requires both a behavioral property and a trace property.

\paragraph{Contributions} 
\begin{itemize}
    \item we propose a semantic model of interacting cyber and physical processes based on sequences of observations,
    \item we define an algebraic framework to express interactions among time sensitive components,
    \item we give a general mechanism, using a co-inductive construction, to define algebraic operations on components as a lifting of some constraints on observations,
    \item we introduce two classes of properties on components, trace properties and behavior properties, and demonstrate their application in an example.
\end{itemize}
This work extends~\cite{LAT21} by (1) including a proof for each result, (2) adding several new results about properties of parametrized products on components, and (3) extending the set of examples that use our model.

Our approach differs from more concrete approaches (e.g., operational models, executable specifications, etc.) in the sense that our operations on components model operations of composition at the semantic level.

We first intuitively introduce some key concepts and an example in Section~\ref{section:problem}. We provide in Section~\ref{section:components} formal definitions for components, their composition, and their properties.
We describe a detailed example in Section~\ref{section:formal-example}.  
We present some related work and our future work in Section~\ref{section:related-work}, 
and conclude the paper in Section~\ref{section:conclusion}.

\newcommand{\ssec}{\si{\second}}
\newcommand{\met}{\si{\meter}}
\newcommand{\speed}{\si{\meter \per \second}}
\section{Coordination of energy-constrained robots on a field}\label{section:problem}

In this work, we consider a cyber-physical system as a set of interacting processes. 
Whether a process consists of a physical phenomenon (sun rising, electro-chemical reaction, etc.) or a cyber phenomenon (computation of a function, message exchanges, etc.), it exhibits an externally observable behavior resulting from some internal non-visible actions.
Instead of a unified way to describe internals of cyber and physical processes, we propose a uniform description of what we can externally observe of their behavior and interactions.

In this section, we introduce some concepts that we will formalize later.
An \emph{event} may describe something like \emph{the sun-rise} or \emph{the temperature reading of 5$^{\circ}C$}.
An event occurs at a point in time, yielding an event occurrence (e.g., the sun-rise event occurred at 6:28 am today), and the same event can occur repeatedly at different times (the sun-rise event occurs every day). 
Typically, multiple events may occur at ``the same time" as measured within a measurement tolerance (e.g., the bird vacated the space at the same time as the bullet arrived there; the red car arrived at the middle of the intersection at the same time as the blue car did).
We call a set of events that occur together at the same time an \emph{observable}. 
A pair $(O, t)$ of a set of observable events $O$ together with its time-stamp $t$ represents an \emph{observation}.
An observation $(O,t)$ in fact consists of a set of event occurrences: occurrences of events in $O$ at the same time $t$.
We call an infinite sequence of observations a \emph{Timed-Event Stream} (TES). 
A \emph{behavior} is a set of TESs.  
A \emph{component} is a behavior  with an interface.

Consider two robot components, each interacting with its own local battery component, sharing a field resource. 
The fact that the robots share the field through which they roam, forces them to somehow coordinate their (move) actions. Coordination is a set of constraints imposed on the otherwise possible observable behavior of components. In the case of our robots, if nothing else, at least physics prevents the two robots from occupying the same field space at the same time. More sophisticated coordination may be imposed (by the robots themselves or by some other external entity) to restrict the behavior of the robots and circumvent some undesirable outcomes, including hard constraints imposed by the physics of the field.
The behaviors of components consist of timed-event streams, where events may include some measures of physical quantities.
We give in the sequel a detailed description of three components, a robot (R), a battery (B), and a field (F), and of their interactions. 
We use SI system units to quantify physical values, with time in seconds (\si{\second}), charging status in Watt hour (Wh), distance in meters (\si{\meter}), force in newtons (\si{\newton}), speed in meters per second (\si{\meter \per \second}).

A \emph{robot} component, with identifier $R$,\short{performs some measurements on its environment and takes decisions as to where to move. We distinguish two kinds of events in the behavior of a robot component} has two kinds of events: a read event $(\rread(\bat,R);b)$ that measures the level $b$ of its battery or $(\rread(\loc,R);l)$ that obtains its position $l$, and a move event $(\move(R);(d, \alpha))$ when the robot moves in the direction $d$ with energy $\alpha$ (in \si{\watt}).
The TES in the Robot column in Table~\ref{table:TESs} shows a scenario where robot $R$ reads its location and gets the value $(0;0)$ at time $1 \ssec$, then moves north with 20\si{\watt} at time $2 \ssec$, reads its location and gets $(0;1)$ at time $3 \ssec$, and reads its battery value and gets $2000\si{Wh}$ at time $4\ssec$, .... 

A \emph{battery} component, with identifier $B$,\short{encapsulates some physical phenomenon that we use three kinds of events to describe its observables} has three kinds of events: a charge event $(\charge(B);\eta_c)$, a discharge event $(\discharge(B);\eta_d)$, and a read event $(\rread(B);s)$,  where $\eta_d$ and $\eta_c$ are respectively the discharge and charge rates of the battery, and $s$ is the current charge status. 
        The TES in the Battery column in Table~\ref{table:TESs}
        shows a scenario where the battery discharged at a rate of $20 \si{\watt}$ at time $2 \ssec$, and reported its charge-level of $2000\si{Wh}$ at time $4 \ssec$, .... 

        \begin{table}
            \centering
            \caption{Each column displays a segment of a timed-event stream for a robot, a battery, and a field component, where observables are singleton events. For $t \in \Rp$, we use $R(t), B(t)$, and $F(t)$ to respectively denote the observable at time $t$ for the \tes\ in the Robot, the Battery, and the Field column. An explicit empty set is not mandatory if no event is observed.}\label{table:TESs}
        {\fontsize{9}{10}\selectfont 
        \begin{tabular}{l|c|c|c|c}
                        &  Robot ($R$)  \quad \quad &   Battery ($B$)                        & Field ($F$)               & \makecell{Robot-Battery-\\Field}\\
            \hline
                $1\si{\second}$\  &  $\{(\rread(\loc,R);(0;0))\}$               &  $                   $                      & $\{(\loc(I);(0;0))\}    $ & $R(1) \cup F(1)$ \\
                $2\si{\second}$\  &  $\makecell{\{(\move(R); (N,20\si{\watt}))\}}$   &  $\{(\discharge(B);20\si{\watt})\}$   & $\{(\move(I);(N,40\si{\newton}))\}$& $R(2) \cup B(2) \cup F(2)$\\
                $3\si{\second}$\  &  $\{(\rread(\loc,R);(0;1))\}$   &  $                   $      & $\{(\loc(I);(0;1))\}    $ & $R(3) \cup F(3)$\\ 
                $4\si{\second}$\  &  $\{(\rread(\bat,R);2000\si{Wh}) \}$   &  $\{(\rread(B);2000\si{Wh})\}   $  &                           & $R(4) \cup B(4)$\\
     $...$          \  & \multicolumn{1}{c}{$...$}       &  \multicolumn{1}{c}{$...$}  & \multicolumn{0}{c}{$...$} &  \multicolumn{0}{c}{$...$}
        \end{tabular}
        }
        \end{table}

    A \emph{field} component, with identifier $F$,
    \short{encapsulates some physical aspects and we use two kinds of events to describe its observables} 
    has two kinds of events: a position event $(\loc(I);p)$ that obtains the position $p$ of an object $I$, and a move event $(\move(I);(d, F))$ of the object $I$ in the direction $d$ with traction force $F$ (in \si{\newton}).
        The TES in the Field column in Table~\ref{table:TESs} 
        shows a scenario where the field has the object $I$ at location $(0;0)$ at time $1\ssec$, then the object $I$ moves in the north direction with a traction force of 40\si{\newton}  at time $2\ssec$, subsequently to which the object $I$ is at location $(0;1)$ at time $3\ssec$, .... 

        When components interact with each other, 
        in a shared environment, behaviors in their composition must also compose with a behavior of the environment.
        For instance, a battery component may constrain how many amperes it delivers, and therefore restrict the speed of the robot that interacts with it.
        We specify interaction explicitly as an exogenous binary operation that constrains the composable behaviors of its operand components.

        The \emph{robot-battery} interaction imposes that a move event in the behavior of a robot coincides with a discharge event in the behavior of the robot's battery, such that the discharge rate of the battery is proportional to the energy needed by the robot. 
        The physicality of the battery prevents the robot from moving if the energy level of the battery is not sufficient 
        (i.e., such an anomalous TES would not exist in the battery's behavior, and therefore cannot compose with a robot's behavior).
        Moreover, a read event in the behavior of a robot component should also coincide with a read event in the behavior of its corresponding battery component, such that the two events contain the same charge value.

        The \emph{robot-field} interaction imposes that a move event in the behavior of a robot coincides with a move event of an object on the field, such that the traction force on the field is proportional to the energy that the robot put in the move. A read event in the behavior of a robot coincides with a position event of the corresponding robot object on the field, such that the two events contain the same position value.
        Additional interaction constraints may be imposed by the physics of the field.
        For instance, the constraint ``no two robots can be observed at the same location'' would rule out every behavior where the two robots are observed at the same location.

        A \tes\ for the composite Robot-Battery-Field system collects, in sequence, all observations from a \tes\ in a Robot, a Battery, and a Field component behavior, such that at any moment the interaction constraints are satisfied. The column Robot-Battery-Field in Table~\ref{table:TESs} displays the first elements of such a \tes.

\newcommand{\Obs}{\mathbb{O}}
\newcommand{\zip}{\mathit{zip}}
\newcommand{\ind}{\mathit{ind}}
\newcommand{\true}{\mathit{true}}
\newcommand{\ssync}{\mathit{sync}}
\newcommand{\eexcl}{\mathit{excl}}
\newcommand{\excl}{\nparallel_\sqcap}
\newcommand{\Rel}{\mathcal{R}}
\setcounter{footnote}{0}

\section{Components, composition, and properties}\label{section:components}
    The definition of components in this section is similar to but differs from the one defined in~\cite{AR03,L19}.
    Intuitively, a component denotes a set of (infinite) sequences of observations. Whether it is a cyber process or a physical process, our notion of component captures all of its possible sequences of observations. 

    A model of interaction emerges naturally from our component model by relating observation of events from one component to observation of events from another component.
    Moreover, we give a construction to lift constraints on observations to constraints on infinite sequences of observations, and ultimately define, from those interaction constraints, algebraic operations on components.

\subsection{Notations}
An \emph{event} is a simplex (the most primitive form of an) observable element. An event may or may not have internal structure. For instance, the successive ticks of a clock are occurrences of a tick event that has no internal structure; successive readings of a thermometer, on the other hand, constitute occurrences of a temperature-reading event
, each of which has the internal structure of a name-value pair
. Similarly, we can consider successive transmissions by a mobile sensor as occurrences of a structured event, each instance of which includes geolocation coordinates, barometric pressure, temperature, humidity, etc. Regardless of whether or not events have internal structures, in the sequel, we regard events as uninterpreted simplex observable elements. 
\begin{notation}[Events]
    We use $\mathbb{E}$ to denote the universal set of events.
\end{notation}

An \emph{observable} is a set of event occurrences that happen together and
an \emph{observation} is a pair $(O,t)$ of an observable $O$ and a time-stamp $t \in \Rp$.\footnote{Any totally ordered dense set would be suitable as the domain for time (e.g., positive rationals $\mathbb{Q}_+$). For simplicity, we use $\Rp$, the set of real numbers $r\geq 0$ for this purpose.} An observation $(O,t)$ represents an act of atomically observing occurrences of events in $O$ at time $t$. Atomically observing occurrences of events in $O$ at time $t$ means there exists a small $\epsilon \in \Rp$ such that during the time interval $[t-\epsilon, t+ \epsilon]$:
\begin{enumerate}
    \item every event $e \in O$ is observed exactly once\footnote{A finer time granularity, i.e., a smaller $\epsilon$, may reveal some ordering relation on the set of events that occur in the same set of observation.}, and
    \item no event $e \not \in O$ is observed.
\end{enumerate}

We write $\langle s_0, s_1, ..., s_{n-1}\rangle$ to denote a \emph{finite sequence of size $n$} of elements over an arbitrary set $S$, where $s_i \in S$ for $0 \leq i \leq n-1$. The set of all finite sequences of elements in $S$ is denoted as $S^*$.
A \emph{stream}\footnote{The set $\N$ denotes the set of natural numbers $n \geq 0$.} over a domain $S$ is 
a function $\sigma : \N \rightarrow S$.
We use 
$\sigma(i)$ to represent the $i+1^{st}$ element of $\sigma$, and given  a finite sequence $s = \langle s_0, ..., s_{n-1}\rangle$, we write $s \cdot \sigma$ to denote the stream $\tau \in \N \rightarrow S$ such that $\tau(i) = s_i$ for $0 \leq i \leq n-1$ and $\tau(i) = \sigma(i-n)$ for $n\leq i$. We use $\sigma^{(n)}$ to denote the $\textit{n-th}$ derivative of $\sigma$, such that $\sigma^{(n)}(i) = \sigma(i+n)$ for all $i \in \mathbb{N}$.
We use $\sigma'$ as an abbreviation for the first derivative of the stream $\sigma$, i.e., $\sigma' = \sigma^{(1)}$.

A \emph{Timed-Event Stream (TES)} over a set of events $E$ and a set of time-stamps $\Rp$ is a stream $\sigma \in \N \rightarrow (\mathcal{P}(E) \times \Rp)$ where, for every $i\in \mathbb{N}$, let $\sigma(i) = (O_i,t_i)$ and:
\begin{enumerate}
    \item $t_i<t_{i+1}$, [i.e., time monotonically increases] and
    \item for every $n \in \N$, there exists $k \in \N$ such that $t_k > n$ [i.e., time is non-Zeno progressive].
\end{enumerate}

\begin{notation}[Time stream]
We use $OS(\Rp)$ to refer to the set of all monotonically increasing and non-Zeno infinite sequences of elements in $\Rp$.
\end{notation}

\begin{notation}[Timed-Event Stream]
We use $\TES{E}{T}$ to denote the set of all TESs whose observables are subsets of the event set $E$ with elements in $\Rp$ as their time-stamps. 
\end{notation}

Given a sequence $\sigma = \langle(O_0,t_0), (O_1,t_1), (O_2,t_2), ...\rangle \in \TES{E}{T}$, we use the projections $\pr_1(\sigma) \in \N \rightarrow \mathcal{P}(E)$ and $\pr_2(\sigma) \in OS(\Rp)$ to denote respectively the sequence  of observables $\langle O_0, O_1, O_2, ...\rangle$ and the sequence of time stamps $\langle t_0, t_1, t_2, ...\rangle$.

\subsection{Components}
The design of complex systems becomes simpler if such systems can be decomposed into smaller sub-systems that interact with each other.
In order to simplify the design of cyber-physical systems,
we abstract from the internal details of both cyber and physical processes, to expose a uniform semantic model.
As a first class entity, a component encapsulates a behavior (set of \tes s) and an interface (set of events).

    Like existing semantic models, such as time-data streams~\cite{AR03}, time signal~\cite{TSSL13}, or discrete clock~\cite{FLDL18}, we use a dense model of time.
    However, we allow for arbitrary but finite interleavings of observations. 
In addition, our structure of an observation imposes atomicity of event occurrences within an observation. 
Such constraint abstracts from the precise timing of each event in the set, and turns an observation into an all-or-nothing transaction.

\begin{definition}[Component]\label{def:component}
    A component is a tuple $C = (E, L)$ where $E \subseteq \E$ is a set of events, and $L \subseteq \TES{E}{\Rp}$ is a set of \tes s. 
    We call $E$ the \emph{interface} and $L$ the externally observable \emph{behavior} of $C$. 
\end{definition}

In contrast with other component models where observable ranges over the same universal set of events, therefore making component omniscient, our model encapsulates the set of observable events of a component in its interface. Thus, a component \emph{cannot observe} an event that is not in its interface.
Moreover, Definition~\ref{def:component} makes no distinction between cyber and physical components. 
We use the following examples to describe some cyber and physical aspects of components.
\begin{example}
    \label{ex:cyber}
    Consider a set of two events $E = \{0,1\}$, and restrict our observations to $\{1\}$ and $\{0\}$. 
    A component whose behavior contains \tes s with alternating observations of $\{1\}$ and $\{0\}$ is defined by the tuple $(E, L)$ where
    \[
\begin{array}{cc}
    L = \{ \sigma \in \TES{E}{\Rp} \mid \forall i \in \mathbb{N}. &(\pr_1(\sigma)(i) = \{0\} \implies \pr_1(\sigma)(i+1) = \{1\}) \land \\
                                                                   &(\pr_1(\sigma)(i) = \{1\} \implies \pr_1(\sigma)(i+1) = \{0\}) \}
\end{array}
\]
    Note that this component is oblivious to time, and any stream of monotonically increasing non-Zeno real numbers would serve as a valid stream of time stamps for any such sequence of observations.
\hfill $\blacksquare$
\end{example}
\begin{example}\label{ex:physical}
    Consider a component encapsulating a continuous function $f: (D_0 \times \Rp) \rightarrow D$, where $D_0$ is a set of initial values, and $D$ is the codomain of values for $f$. 
    Such a function can describe the evolution of a physical system over time, where $f(d_0, t) = d$ means that at time $t$ the state of the system is described by the value $d \in D$  if initialized with $d_0$.
    We define the set of all events for this component as the range of function $f$ given an initial parameter $d_0 \in D_0$. 
    The component is then defined as the pair $(D, L_f)$ such that:
    \[
    L_f = \{ \sigma \in \TES{D}{\Rp} \mid \exists d_0 \in D_0.\ \forall i \in \N.\ \pr_1(\sigma)(i) = \{f(d_0,\pr_2(\sigma)(i))\} \}
    \]
    Observe that the behavior of this component contains all possible discrete samplings of the function $f$ at monotonically increasing and non-Zeno sequences of time stamp.  Different instances of $f$ would account for various cyber and physical aspects of components. 
    \hfill $\blacksquare$
\end{example}

\subsection{Composition}\label{subsection:composition}
A complex system typically consists of multiple components that interact with each other.
The example in Section~\ref{section:problem} shows three components, a $\robot$, a $\battery$, and a $\field$, where, for instance,
a move observable of a robot must coincide with an accommodating move observable of the field and a discharge observable of its battery.
    The design challenge is to faithfully represent the interactions among involved components, while keeping the description modular, i.e., specify the robot, the battery, and the field as separate, independent, but interacting components. We present in this section a mechanism to describe composability constraints on behavior, and composition operators to construct complex components out of simpler ones.  
    Such construction opens possibilities for modular reasoning both about the interaction among components and about their resulting composite behavior. 
    Moreover, the composition operator (and its composability relation) reflects the type of interaction between two components, and possibly some independent observations on both operands.  
    An alternative is to consider the product between two components to be the set intersection of their behavior, and artificially equate the interface of all components. 
    This approach requires each component to know the universe of events, which is practically infeasible, and each component to be omniscient on what other components are producing, which breaks the encapsulation principle of our component model. It also gives an artificial illusion, at a design level, that every component synchronizes with every other at each iteration.
    Our operators localize interaction to the two operands only, and are parametrized with the interface of such components.

We express composability constraints on behaviors using relations\footnote{Non-binary relations may also be considered, i.e., constraints imposed on more than two components.}. 
We introduce a generalized notion of a \emph{composability relation} to capture the allowed interaction among two components. 
By modeling composability constraints explicitly, we expose the logic of the interaction that governs the formation of a composite behavior between two components.
\begin{definition}[Composability relation on TESs]
    A composability relation is a parametrized relation $R$ such that for all $E_1, E_2 \subseteq \E$, we have $R(E_1,E_2) \subseteq \TES{E_1}{R} \times \TES{E_2}{R}$.
\end{definition}
\begin{definition}[Symmetry]
    \label{def:symmetry}
    A parametrized relation $Q$ is \emph{symmetric} if, for all $(x_1,x_2)$ and for all $(X_1, X_2)$:  $(x_1,x_2) \in Q(X_1,X_2) \iff (x_2,x_1) \in Q(X_2,X_1)$. 
\end{definition}
    A composability relation on \tes s serves as a necessary constraint for two \tes s to compose.
    We give in Section~\ref{subsection:construction} some examples of useful composability relations on \tes s that, e.g., enforce synchronization or mutual exclusion of observables.
    We define \emph{composition} of \tes s as the act of forming a new \tes\ out of two \tes s.
\begin{definition}
    A composition function $\oplus$ on \tes\  is a function $\oplus: \TES{\E}_\times \TES{\E}_ \rightarrow\TES{\E}{}$.
\end{definition}
We define a binary product operation on components, parametrized by a composability relation and a composition function on \tes s, that forms a new component. Intuitively, the newly formed component describes, by its behavior, the evolution of the joint system under the constraint that the interactions in the system satisfy the composability relation. 
Formally, the product operation returns another component, whose set of events is the union of sets of events of its operands, and its behavior is obtained by composing all pairs of \tes s in the behavior of its operands deemed composable by the composability relation.

\begin{definition}[Product]\label{def:prod}
    Let $(R,\oplus)$ be a pair of a composability relation  and a composition function on \tes s, and $C_i = (E_i, L_i)$, $i \in \{1,2\}$, two components.
    The product of $C_1$ and $C_2$, under $R$ and $\oplus$, denoted as $C_1 \times_{(R,\oplus)} C_2$, is the component $(E, L)$ where $E = E_1 \cup E_2$ and $L$ is defined by 
    \[
    L = \{ \sigma_1 \oplus \sigma_2 \mid \sigma_1 \in L_1,\ \sigma_2 \in L_2,\  (\sigma_1, \sigma_2) \in R(E_1,E_2)\}
    \]
\end{definition}
Definition~\ref{def:prod} presents a generic composition operator, where composition is parametrized over a composability relation and a composition function.
    The product of two components indirectly depends on the interface of its operands, since its composability relation does so.
    Therefore, it is \emph{a priori} not certain that algebraic properties such as commutativity or associativity hold for such user defined products.
    We give in Lemma~\ref{lemma:prod} sufficient conditions on a composability relation and a composition function for a product to be associative, commutative, and idempotent.\footnote{Distributivity holds for some products. We leave the study of the conditions under which distributivity holds as future work.}

\begin{lemma}
    \label{lemma:prod}
    Let  $\oplus$ be a composition function on \tes s, and let $R$ be a composability relation on \tes s.
    Then:
    \begin{itemize}
        \item if $R$ is symmetric, then $\times_{(R, \oplus)}$ is \emph{commutative} if and only if $\oplus$ is commutative; 
        \item if $R$ is such that, for all $E_1, E_2, E_3 \subseteq \E$,
        \[
            (\sigma_1, \sigma_2 \oplus \sigma_3) \in R(E_1, E_2 \cup E_3) \land (\sigma_2, \sigma_3)\in R(E_2, E_3) \iff (\sigma_1, \sigma_2) \in R(E_1, E_2) \land (\sigma_1 \oplus \sigma_2, \sigma_3) \in R(E_1 \cup E_2, E_3)
        \]
    then $\times_{(R,\oplus)}$ is \emph{associative} if and only if $\oplus$ is associative;
       \item if for all $E \subseteq \E$ and $\sigma,\tau \in \TES{E}{\Rp}$, we have $(\sigma, \tau) \in R(E, E) \implies \sigma = \tau$, then $\times_{(R,\oplus)}$ is \emph{idempotent} if and only if $\oplus$ is idempotent.
    \end{itemize}
\end{lemma}
\begin{proof}
    \underline{Commutativity.} 
    Let $C_1 = (E_1, L_1)$ and $C_2 = (E_2, L_2)$ be two components, and $(R,\oplus)$ be a pair of a composability relation and composition function on \tes s. 
    We write $C = (E,L) = C_1 \times_{(R,\oplus)} C_2$ and $C' = (E', L') = C_2 \times_{(R, \oplus)} C_1$.
    We first observe that $E = E_1 \cup E_2 = E'$.
    The condition for the product of two components to be commutative reduces to showing that $L = L'$, also equivalently written as: 
    \[
    \begin{array}{rl}
        L = L'\iff  &\{\sigma_1 \oplus \sigma_2 \mid \sigma_1 \in L_1,\ \sigma_2 \in L_2,\ (\sigma_1,\sigma_2) \in R(E_1, E_2)\} \\
    & = \{\sigma_2 \oplus \sigma_1 \mid \sigma_1 \in L_1,\ \sigma_2 \in L_2,\ (\sigma_2,\sigma_1) \in R(E_2, E_1)\}
    \end{array}
    \]
    If $R$ is symmetric (as in Definition~\ref{def:symmetry}) and $\oplus$ is commutative, then  
    $L = L'$. Hence, if $R$ is symmetric and $\oplus$ is commutative, then $\times_{(R,\oplus)}$ is commutative.

    Oppositely, if $R$ is symmetric and $L = L'$, we show that $\oplus$ is commutative.
    We take the symmetric relation $R(E_1, E_2)$ such that $(\sigma_1, \sigma_2) \in R(E_1, E_2)$ for any $\sigma_1 \in \TES{E_1}{}$ and $\sigma_2 \in \TES{E_2}{}$.
    Let $C_\sigma$ be the component $(E_\sigma, \{\sigma\})$ where $E_\sigma = \bigcup \{\sigma(i) \mid i \in \mathbb{N}\}$.
    Thus, for any $\sigma_1\in \TES{E_1}{}$ and $\sigma_2 \in \TES{E_2}{}$, $C_{\sigma_1} \times_{(R,\oplus)} C_{\sigma_2} = (E_{\sigma_1} \cup E_{\sigma_2}, \{\sigma_1 \oplus \sigma_2\}$. 
    A necessary condition for $\times_{(R,\oplus)}$ to be commutative is that $\{\sigma_1 \oplus \sigma_2\} = \{\sigma_2 \oplus \sigma_1\}$, which imposes commutativity on $\oplus$.
    \medskip

    \underline{Associativity.} 
    Let $(R,\oplus)$ be a pair of a composability relation and a composition function on \tes s. 
    We consider three components $C_i = (L_i,E_i)$, with $i \in \{1,2,3\}$.

    The set of events for component $((C_1 \times_{(R,\oplus)} C_2) \times_{(R,\oplus)} C_3)$ is the set $E_1 \cup E_2 \cup E_3$, which is equal to the set of events of component $(C_1 \times_{(R,\oplus)} (C_2 \times_{(R,\oplus)} C_3))$.

Let $L'$ and $L''$ respectively be the behaviors of components $(C_1 \times_{(R,\oplus)} C_2) \times_{(R,\oplus)} C_3$ and $C_1 \times_{(R,\oplus)} (C_2 \times_{(R,\oplus)} C_3)$. 
We show some sufficient conditions for $L' = L''$, also written as 
\[
    \begin{array}{l}
        \{ (\sigma_1\oplus\sigma_2)\oplus\sigma_3 \mid 
                \sigma_1 \in L_1,\ \sigma_2 \in L_2,\ \sigma_3 \in L_3. 
                \ (\sigma_1,\sigma_2) \in R (E_1,E_2)  \land \\
               \qquad \qquad \qquad (\sigma_1\oplus \sigma_2,\sigma_3) \in R (E_1 \cup E_2,E_3)
\} \\
            =
        \{ \sigma_1\oplus(\sigma_2\oplus\sigma_3) \mid 
                \sigma_1 \in L_1,\ \sigma_2 \in L_2,\ \sigma_3 \in L_3.
                \ (\sigma_2,\sigma_3) \in R (E_2,E_3)  \land  \\
            \qquad \qquad \qquad  (\sigma_1, \sigma_2 \oplus \sigma_3) \in R (E_1,E_2\cup E_3)\}
    \end{array}
\]

We first observe that if $\oplus$ is associative, then a sufficient condition for $L'$ to be equal to $L''$ is that 
\[
    \begin{array}l
        (\sigma_1,\sigma_2) \in R (E_1,E_2)  \land (\sigma_1\oplus \sigma_2,\sigma_3) \in R (E_1 \cup E_2,E_3) \iff\\
(\sigma_2,\sigma_3) \in R (E_2,E_3)  \land (\sigma_1, \sigma_2 \oplus \sigma_3) \in R (E_1,E_2\cup E_3)
    \end{array}
\]
    for every $(\sigma_1, \sigma_2, \sigma_3) \in L_1\times L_2 \times L_3$.
   Assuming that  $R$ satisfies the following constraint for every $(\sigma_1, \sigma_2, \sigma_3) \in L_1\times L_2 \times L_3$:
   \[
    \begin{array}l
        (\sigma_1,\sigma_2) \in R (E_1,E_2)  \land (\sigma_1\oplus \sigma_2,\sigma_3) \in R (E_1 \cup E_2,E_3) \iff\\
(\sigma_2,\sigma_3) \in R (E_2,E_3)  \land (\sigma_1, \sigma_2 \oplus \sigma_3) \in R (E_1,E_2\cup E_3)
    \end{array}
\]
    and suppose that $C_{\sigma_1} \times_{(R, \oplus)} (C_{\sigma_2} \times_{(R, \oplus)} C_{\sigma_3}) = (C_{\sigma_1} \times_{(R, \oplus)} C_{\sigma_2}) \times_{(R, \oplus)} C_{\sigma_3}$, then $\sigma_1 \oplus (\sigma_2 \oplus \sigma_3) = (\sigma_1 \oplus \sigma_2) \oplus \sigma_3$ which implies that $\oplus$ is associative. 
    Thus, if $R$ satisfies the constraint as written above, $\times_{(R,\oplus)}$ is associative if and only if $\oplus$ is associative.
    \medskip

    \underline{Idempotency.}
    We show that if for all $E \subseteq \E$, and $\sigma,\tau \in \TES{E}{}$, we have that $(\sigma, \tau) \in R (E, E)$ implies $\sigma = \tau$, then $\times_{(R,\oplus)}$ is idempotent if and only if $\oplus$ is idempotent.
    We first observe that, given a component $C = (E,L)$, the component $C \times_{(R,\oplus)} C = (E,L')$ has the same set of events, $E$.

    We show that $(\sigma_1,\sigma_2) \in R (E,E) \implies \sigma_1 = \sigma_2$ and $\oplus$ idempotent is a sufficient condition for having $L' = L$.  
    Indeed,
    \begin{align*}
        L' =\ &\{ \sigma_1 \oplus \sigma_2 \mid \sigma_1,\sigma_2 \in L,\ (\sigma_1, \sigma_2) \in R (E, E)  \}\\
         =\ &\{ \sigma_1 \oplus \sigma_1 \mid \sigma_1 \in L \}\\
                            =\ & L
    \end{align*}

    Similar to the previous cases, if for all $E \subseteq \E$, and $\sigma,\tau \in \TES{E}{}$, we have that $(\sigma, \tau) \in R (E, E)$ implies $\sigma = \tau$, then $\times_{(R,\oplus)}$ is idempotent if and only if $\oplus$ is idempotent.
    The proof is similar to the case of commutativity, i.e. showing that $C_{\sigma} \times_{(R, \oplus)} C_{\sigma}  = C_{\sigma}$ implies that $\sigma \oplus \sigma = \sigma$ for all $\sigma \in \TES{E}{}$. 
\end{proof}

The generality of our formalism allows exploration of other kinds of operations on components, such as division.
Intuitively, the division of a component $C_1$ by a component $C_2$ yields a component $C_3$ whose behavior contains all \tes s that can compose with \tes s in the behavior of $C_2$ to yield the \tes s in the behavior of $C_1$.
\begin{definition}[Division]
    \label{def:div}
    Let $R$ be a composability relation on \tes s, and $\oplus$ a composition function on \tes s.
    The division of two components $C_1 =(E_1, L_1)$ and $C_2 = (E_2, L_2)$ under $R$ and $\oplus$, denoted as $C_1 /_{(R,\oplus)} C_2$, is the component $C = (E_1, L)$ such that: 
    \[ L = \{ \sigma \in \TES{E_1}{\Rp} \mid \exists \sigma_2 \in L_2.\ (\sigma, \sigma_2) \in R (E_1, E_2) \land \sigma \oplus \sigma_2 \in L_1 \}\]
\end{definition}
If the dividend is $C_1 = C'_1 \times_{(R,\oplus)} C'_2$, and the divisor is an operand of the product, e.g., $C_2 = C'_2$, then the behavior of the result of the division, $C$, contains all \tes s in the behavior of the other operand (i.e., $C'_1$) composable with a \tes\ in the behavior of $C_2$.
\begin{lemma}\label{lemma:division}
    Let $C_1 = (E_1, L_1)$ and $C_2= (E_2, L_2)$ be two components.  
    Let $(C_1 \times_{(R, \oplus)} C_2) /_{(R, \oplus)} C_2 = (E_3,L_3)$, with $(R,\oplus)$ a pair of a  composability relation and a composition function on \tes s.
Then, \[\{\sigma_1 \in L_1 \mid \exists \sigma_2 \in L_2.\ (\sigma_1,\sigma_2) \in R(E_1 \cup E_2, E_2) \cap R(E_1, E_2)\} \subseteq L_3\]
\end{lemma}
\begin{proof}
    Let $(R,\oplus)$ be a pair of a composability relation and a composition function on \tes s.
    For any components $C_1 = (E_1, L_1)$ and $C_2= (E_2, L_2)$, let $(C_1 \times_{(R, \oplus)} C_2) = (E_1 \cup E_2,L')$ and $(C_1 \times_{(R, \oplus)} C_2) /_{(R, \oplus)} C_2 = (E,L)$.
    The set $L$ is such that, for any $\sigma \in \TES{E_1 \cup E_2}{}$:
    \[ \sigma \in L \iff \exists \sigma_2 \in L_2, \tau \in L'.\ (\sigma, \sigma_2) \in R (E_1 \cup E_2, E_2) \land \tau = (\sigma \oplus \sigma_2) \]
    By construction, the existence of $\tau \in L'$ is equivalent to the existence of $\tau_1 \in L_1$ and $\tau_2 \in L_2$ such that 
\[ (\tau_1, \tau_2) \in R (E_1, E_2) \land \tau = (\tau_1 \oplus \tau_2) \]
    Thus, for any $\sigma \in \TES{E_1 \cup E_2}{}$:
    \setcounter{equation}{0}
    \begin{align*}
        \hspace{-1em}\sigma \in L &\iff \exists \tau_1 \in L_1, \tau_2 \in L_2, \sigma_2\in L_2.\\
                                  & \qquad\ \ \ (\sigma, \sigma_2) \in R (E_1 \cup E_2, E_2) \land \sigma\oplus \sigma_2 = \tau_1 \oplus \tau_2 \land (\tau_1, \tau_2) \in R (E_1, E_2)\\
                                  &\impliedby \sigma \in L_1 \land 
                               \exists \sigma_2\in L_2.
                               (\sigma, \sigma_2) \in R(E_1 \cup E_2, E_2) \land (\sigma, \sigma_2) \in R (E_1, E_2)
    \end{align*}
    The last implication concludes the proof that 
\[\{\sigma \in L_1 \mid 
                               \exists \sigma_2\in L_2.
                       (\sigma, \sigma_2) \in R(E_1 \cup E_2, E_2) \land (\sigma, \sigma_2) \in R (E_1, E_2)\} \subseteq L\]
\end{proof}
\begin{corollary}
    Let $C_1 = (E_1, L_1)$ and $C_2= (E_2, L_2)$ be two components.  
    Let $(C_1 \times_{(\top, \oplus)} C_2) /_{(\top, \oplus)} C_2 = (E_3,L_3)$, with $\oplus$ a composition function on \tes s (see Definition~\ref{def:cross} for $\top$).
    Then $L_1 \subseteq L_3$.
\end{corollary}
\begin{proof}
    Using Lemma~\ref{lemma:division},
\[\{\sigma_1 \in L_1 \mid \exists \sigma_2 \in L_2.\ (\sigma_1,\sigma_2) \in \top\} 
 \subseteq L_3 
\implies L_1 
 \subseteq L_3 
 \]
\end{proof}
    The results of Section~\ref{subsection:composition} show a wide variety of product operations that our semantic model offers. 
    Given a fixed set of components, one can change how components interact by choosing different composability relations and composition functions. 
    We also give some sufficient conditions for a product operation on components to be associative, commutative, and idempotent, in terms of the algebraic properties of its composability relation and its composition function. Such results are useful to simplify and to prove equivalent two component expressions.

\subsection{A co-inductive construction for composition operators}\label{subsection:construction}
In Section~\ref{subsection:composition}, we presented a general framework to design components in interaction.
As a result, the same set of components, under different forms of interaction, leads to the creation of alternative systems.
The separation of the composability constraint and the composition operation gives complete control to design different interaction protocols among components.

In this section, we provide a co-inductive construction for composability relations on \tes s.
We show how constraints on observations can be \emph{lifted} to constraints on \tes s, and give weaker conditions for Lemma~\ref{lemma:prod} to hold.
    The intuition for such construction is that, in some cases, the condition for two \tes s to be composable depends only on a composability relation on observations. An example of composability constraint for a robot with its battery and a field enforces that each \emph{move} event \emph{discharges} the battery and \emph{changes} the state of the field. As a result, every \emph{move} event observed by the robot must coincide with a \emph{discharge} event observed by the battery and a change of state observed by the field. The lifting of such composability relation on observations to a constraint on \tes s is defined co-inductively.

\begin{definition}[Composability relation on observations]
    A composability relation on observations is a parametrized relation $\kappa$ such that for all pairs
    $(E_1, E_2) \in \Po(\E) \times \Po(\E) $, we have $\kappa(E_1,E_2) \subseteq (\Po(E_1) \times \Rp) \times (\Po(E_2) \times \Rp)$
\end{definition}
For two composability relations $\kappa_1, \kappa_2$, their intersection or union, written $\kappa_1 \cap \kappa_2$ and $\kappa_1 \cup \kappa_2$ respectively, is defined, for any $E_1, E_2, E_3\subseteq \E$, as $(\kappa_1 \cap \kappa_2) (E_1, E_2) = \kappa_1(E_1, E_2) \cap \kappa_2(E_1, E_2)$ and $(\kappa_1 \cup \kappa_2)(E_1, E_2) = \kappa_1(E_1, E_2) \cup \kappa_2(E_1, E_2)$.
\begin{definition}[Lifting- composability relation]
    Let $\kappa$ be a composability relation on observations, and let $\Phi_\kappa(E_1, E_2) : (\Po(\TES{E_1}{\Rp}) \times \Po(\TES{E_2}{\Rp})) \rightarrow (\Po(\TES{E_1}{\Rp}) \times \Po(\TES{E_2}{\Rp}))$ be such that, for any $\Rel \subseteq \TES{E_1}{} \times \TES{E_2}{}$:
 \[
\begin{array}{rl}
    \Phi_\kappa(E_1,E_2)(\Rel) = \{(\tau_1, \tau_2) \mid & (\tau_1(0), \tau_2(0)) \in \kappa(E_1,E_2) \land   \\
                                                         & (\pr_2(\tau_1)(0) =t_1 \land  \pr_2(\tau_2)(0) =t_2) \land \\
                                                         &(t_1 < t_2 \land (\tau_1',\tau_2) \in \Rel \lor t_2 < t_1 \land (\tau_1,\tau_2') \in \Rel \lor t_2 = t_1 \land (\tau_1',\tau_2') \in \Rel)\}
\end{array}
\]
The \emph{lifting} of $\kappa$ on \tes s, written $[\kappa]$, is the parametrized relation obtained by taking the greatest post fixed point of the function $\Phi_\kappa(E_1, E_2)$ for arbitrary pair $E_1, E_2 \subseteq \E$, i.e., the relation $[\kappa](E_1, E_2) = \bigcup_{\Rel \subseteq \TES{E_1}{} \times \TES{E_2}{}} \{ \Rel \mid \Rel \subseteq \Phi_\kappa(E_1, E_2)(\Rel)\}$.
\end{definition}

\begin{lemma}[Correctness of lifting]\label{lemma:correctness-lifting}
    For any $E_1, E_2\subseteq \E$, the function $\Phi_\kappa(E_1, E_2)$ is monotone, and therefore has a greatest post fixed point.
\end{lemma}
\begin{proof}
Let $\kappa$ be a composability relation on observations, and let $E_1, E_2\subseteq \E$.
We recall that the function $\Phi_\kappa(E_1, E_2)$ is such that, for any $\Rel \subseteq \TES{E_1}{} \times \TES{E_2}{}$:
     \[
\begin{array}{rl}
    \Phi_\kappa(E_1,E_2)(\Rel) = \{(\tau_1, \tau_2) \mid & (\tau_1(0), \tau_2(0)) \in \kappa(E_1,E_2) \land   \\
                                                         & (\pr_2(\tau_1)(0) =t_1 \land  \pr_2(\tau_2)(0) =t_2) \land \\
                                                         &(t_1 < t_2 \land (\tau_1',\tau_2) \in \Rel \lor t_2 < t_1 \land (\tau_1,\tau_2') \in \Rel \lor  t_2 = t_1 \land (\tau_1',\tau_2') \in \Rel)\}
\end{array}
\]

Let $\Rel_1, \Rel_2 \subseteq \TES{E_1}{} \times \TES{E_2}{}$ be such that $\Rel_1 \subseteq \Rel_2$. We show that $\Phi_\kappa(E_1, E_2)(\Rel_1) \subseteq \Phi_\kappa(E_1, E_2)(\Rel_2)$.
For any $(\tau_1,\tau_2) \in \TES{E_1}{}\times \TES{E_2}{}$,
    \[
\begin{array}{rl}
    (\tau_1,\tau_2) \in \Phi_\kappa(E_1, E_2)(\Rel_1) \iff& (\tau_1(0), \tau_2(0)) \in \kappa(E_1,E_2) \land   \\
                                                         & (\pr_2(\tau_1)(0) =t_1 \land  \pr_2(\tau_2)(0) =t_2) \land \\
                                                         &(t_1 < t_2 \land (\tau_1',\tau_2) \in \Rel_1 \lor t_2 < t_1 \land (\tau_1,\tau_2') \in \Rel_1 \lor  \\
                                                         &\ t_2 = t_1 \land (\tau_1',\tau_2') \in \Rel_1)\\
    \implies& (\tau_1(0), \tau_2(0)) \in \kappa(E_1,E_2) \land   \\
                                                         & (\pr_2(\tau_1)(0) =t_1 \land  \pr_2(\tau_2)(0) =t_2) \land \\
                                                         &(t_1 < t_2 \land (\tau_1',\tau_2) \in \Rel_2 \lor t_2 < t_1 \land (\tau_1,\tau_2') \in \Rel_2 \lor  \\
                                                         &\ t_2 = t_1 \land (\tau_1',\tau_2') \in \Rel_2)\\
    \implies & (\tau_1, \tau_2) \in \Phi_\kappa(E_1, E_2)(\Rel_2)
\end{array}
\]

Therefore, $\Rel_1 \subseteq \Rel_2$ implies that $\Phi_\kappa(E_1, E_2)(\Rel_1) \subseteq \Phi_\kappa(E_1, E_2)(\Rel_2)$, and we conclude that $\Phi_\kappa(E_1, E_2)$ is monotonic. By the greatest fixed point theorem, $\Phi_\kappa(E_1, E_2)$ has a greatest fixed point defined as:
\[
[\kappa](E_1, E_2) = \bigcup \{ \Rel \mid \Rel \subseteq \Phi_\kappa(E_1, E_2)(\Rel)\}
\]
\end{proof}
\begin{lemma}\label{lemma:lifting}
    If $\kappa$ is a composability relation on observations, then the lifting $[\kappa]$  is a composability relation on \tes s.
    Moreover, if $\kappa$ is symmetric (as in Definition~\ref{def:symmetry}), then $[\kappa]$ is symmetric.
\end{lemma}
\begin{proof}
    We first note that, given a composability relation $\kappa$ on observables, the lifting $[\kappa]$ is a composability relation on \tes s.
    Indeed, for any pair of interfaces $E_1, E_2 \subseteq \E$, any $(\sigma, \tau) \in [\kappa](E_1, E_2)$ is a pair in $\TES{E_1}{} \times \TES{E_2}{}$.

    If $\kappa$ is symmetric (as in Definition~\ref{def:symmetry}), we show that $[\kappa]$ is also symmetric.
    Given a set $\Rel \subseteq \TES{E_1}{} \times \TES{E_2}{}$, we use the notation $\overbar{\Rel}$ to denote the smallest set such that $(\sigma,\tau) \in \Rel \iff (\tau, \sigma) \in \overbar{\Rel}$.
    Let $E_1, E_2 \subseteq \E$.

    If $\kappa$ is symmetric, then for $\Rel \subseteq \TES{E_1}{}\times\TES{E_2}{}$,
    \[
\begin{array}{rl}
    \Phi_\kappa(E_1,E_2)(\Rel) = \{(\tau_1, \tau_2) \mid & (\tau_1(0), \tau_2(0)) \in \kappa(E_1,E_2) \land   \\
                                                         & (\pr_2(\tau_1)(0) =t_1 \land  \pr_2(\tau_2)(0) =t_2) \land \\
                                                         &(t_1 < t_2 \land (\tau_1',\tau_2) \in \Rel \lor t_2 < t_1 \land (\tau_1,\tau_2') \in \Rel \lor  \\
                                                         &\ t_2 = t_1 \land (\tau_1',\tau_2') \in \Rel)\}\\
= \{(\tau_1, \tau_2) \mid & ( \tau_2(0), \tau_1(0)) \in \kappa(E_2,E_1) \land   \\
                                                         & (\pr_2(\tau_1)(0) =t_1 \land  \pr_2(\tau_2)(0) =t_2) \land \\
                                                         &(t_1 < t_2 \land (\tau_1',\tau_2) \in \Rel \lor t_2 < t_1 \land (\tau_1,\tau_2') \in \Rel \lor  \\
                                                         &\ t_2 = t_1 \land (\tau_1',\tau_2') \in \Rel)\}\\
= \{(\tau_1, \tau_2) \mid & ( \tau_2(0), \tau_1(0)) \in \kappa(E_2,E_1) \land   \\
                                                         & (\pr_2(\tau_1)(0) =t_1 \land  \pr_2(\tau_2)(0) =t_2) \land \\
                                                         &(t_1 < t_2 \land (\tau_2, \tau_1') \in \overbar{\Rel} \lor t_2 < t_1 \land (\tau_2', \tau_1) \in \overbar{\Rel} \lor  \\
                                                         &\ t_2 = t_1 \land (\tau_2', \tau_1') \in \overbar{\Rel})\}\\
    = \{(\tau_1, \tau_2) \mid & ( \tau_2, \tau_1) \in \Phi_\kappa(E_2, E_1)(\overbar{\Rel})\} \hfill (1)
\end{array}
\]
which shows that $[\kappa]$ is symmetric since, for any $E_1, E_2 \subseteq \E$, 
$[\kappa](E_1, E_2) = 
\bigcup_{\Rel \subseteq \TES{E_1}{} \times \TES{E_2}{}} \{ \Rel \mid \Rel \subseteq \Phi_\kappa(E_1, E_2)(\Rel)\}$, and 
\[
\begin{array}{rl}
    (\sigma, \tau) \in [\kappa](E_1, E_2) & \iff \exists \Rel.\ (\sigma, \tau) \in \Rel \land \Rel \subseteq \Phi_\kappa(E_1, E_2)(\Rel) \\
                                          & \iff \exists \overbar{\Rel}.\ (\tau, \sigma) \in \overbar{\Rel} \land \overbar{\Rel} \subseteq \Phi_\kappa(E_2, E_1)(\overbar{\Rel}) \\
                                          & \iff (\tau, \sigma) \in [\kappa](E_2, E_1) 
\end{array}
\]
where the first equivalence is given by the fact that $[\kappa](E_1, E_2)$ is the greatest post fixed point of $\Phi_\kappa(E_1, E_2)$, the second equivalence is obtained from equality $(1)$, and the third equivalence is given by the fact that $[\kappa](E_2, E_1)$ is the greatest post fixed point.

\end{proof}

As a consequence of Lemma~\ref{lemma:lifting}, any composability relation on observations gives rise to a composability relation on \tes s.
   We define three composability relations on \tes s, where Definition~\ref{def:sync} and Definition~\ref{def:excl} are two examples that construct co-inductively the composability relation on \tes s from a composability relation on observations. 
For the following definitions, let $C_1 = (E_1, L_1)$ and $C_2 = (E_2, L_2)$ be two components, and $\oplus$ be a composition function on \tes s. 
We use $\sqcap \subseteq \Po(\E) \times \Po(\E)$ to range over relations on observables.

\begin{definition}[Free composition]\label{def:cross}
    We use $\top$ for the most permissive composability relation on \tes s such that, for any $E_1, E_2 \subseteq \E$ and any $\sigma\in\TES{E_1}{}$ and $\tau \in \TES{E_2}{}$, then $(\sigma,\tau) \in \top(E_1, E_2)$.
    \hfill $\blacksquare$
\end{definition}
    The behavior of component $C_1 \times_{(\top, \oplus)} C_2$ contains every \tes\ obtained from the composition under $\oplus$ of every pair $\sigma_1 \in L_1$ and $\sigma_2 \in L_2$ of \tes s. 
              This product does not impose any constraint on event occurrences of its operands.

\begin{definition}[Synchronous composition]\label{def:sync}
        Let $\sqcap \subseteq \Po(\E)^2$ be a relation on observables. 
        We say that two observations are synchronous under $\sqcap$ if, intuitively, the two following conditions hold:
        \begin{enumerate}
            \item every observable that can compose (under $\sqcap$) with another observable must occur simultaneously with one of its related observables; and
            \item only an observable that does not compose (under $\sqcap$) with any other observable can happen before another observable, i.e., at a strictly lower time.
        \end{enumerate}
        To formalize the conditions above, we use the independence relation $\ind_\sqcap(X,Y) = \forall x \subseteq X.\forall y \subseteq Y. (x,y) \not\in \sqcap$.

    The \emph{synchronous} composability relation on observations $\kappa^{\ssync}_{\sqcap}(E_1,E_2)$ is the smallest set such that:
    \begin{itemize}
        \item for all $(O_1, O_2) \in \Po(E_1)\times \Po(E_2)$ such that $(O_1, O_2) \in \sqcap$ and for all $(O_1', O_2') \in \Po(E_1)\times \Po(E_2)$ such that $\ind_\sqcap(O_1', E_2)$ and $\ind_\sqcap(E_1, O_2')$ then, for all time stamps $t$, $((O_1\cup O_1', t), (O_2 \cup O_2',t)) \in \kappa_\sqcap^{\ssync}(E_1, E_2)$.
        \item  if $\ind_\sqcap(O_1, E_2)$ then for all $O_2 \subseteq E_2$ and $t_1 \leq t_2$, $((O_1,t_1), (O_2,t_2)) \in \kappa^\ssync_\sqcap(E_1,E_2)$. Reciprocally, if $\ind_\sqcap(E_1, O_2)$ then for all $O_1 \subseteq E_1$ and $t_2 \leq t_1$, $((O_1,t_1), (O_2,t_2)) \in \kappa^\ssync_\sqcap(E_1,E_2)$;
    \end{itemize}
\end{definition}
\begin{example}
    Let $E_1 = \{a,b\}$ and $E_2 = \{c,d\}$ with $\sqcap = \{(\{a\},\{c\})\}$.
Thus, $((\{a\}, t_1), (\{d\},t_2)) \in \kappa^{\ssync}_{\sqcap}$ if and only if $t_2 < t_1$.
        Alternatively, we have $((\{a\}, t_1), (\{c\},t_2)) \in \kappa^{\ssync}_{\sqcap}$ if and only if $t_1 = t_2$.
    \hfill $\blacksquare$
\end{example}

The behavior of component $C_1 \times_{([\kappa^{\ssync}_{\sqcap}], \oplus)} C_2$ contains \tes s obtained from the composition under $\oplus$ of every pair $\sigma_1 \in L_1$ and $\sigma_2 \in L_2$ of \tes s that are related by the synchronous composability relation $[\kappa^{\ssync}_{\sqcap}]$ which, depending on $\sqcap$, may exclude some event occurrences unless they synchronize.
\footnote{If we let $\oplus$ be the element wise set union, define an event as a set of port assignments, and in the pair $([\kappa^{\ssync}_{\sqcap}], \oplus)$ let $\sqcap$ be true if and only if all common ports get the same value assigned, then this composition operator produces results similar to the composition operation in Reo~\cite{AR03}.}

\begin{definition}[Mutual exclusion]\label{def:excl}
    Let $\sqcap \subseteq \Po(\E)^2$ be a relation on observables. 
We define two observations to be mutually exclusive under the relation $\sqcap$ if no pair of observables in $\sqcap$ can be observed at the same time.
The mutually exclusive composability relation $\kappa^{\eexcl}_{\sqcap}$ on observations allows the composition of two observations $(O_1, t_1)$ and $(O_2, t_2)$, i.e., $((O_1, t_1), (O_2, t_2)) \in \kappa^{\eexcl}_{\sqcap}(E_1, E_2)$, if and only if
\[
    (t_1 \not = t_2) \lor   
    (t_1 = t_2 \land \neg(O_1 \sqcap O_2))
\]
\end{definition}
\begin{example}
    Let $E_1 = \{a,b\}$ and $E_2 = \{c,d\}$ with $\sqcap = \{(\{a\},\{c\})\}$.
Thus, $((\{a\}, t_1), (\{c\},t_2)) \not \in \kappa^{\eexcl}_{\sqcap}$ for any $t_2 = t_1$, since $\{a\}$ and $\{c\}$ are two mutually exclusive observables.
\end{example}

    The behavior of component $C_1 \times_{([\kappa^{\eexcl}_{\sqcap}], \oplus)} C_2$ contains \tes s resulting from the composition under $\oplus$ of every pair $\sigma_1 \in L_1$ and $\sigma_2 \in L_2$ of \tes s that are related by the mutual exclusion composability relation $[\kappa^{\eexcl}_{\sqcap}]$ which, depending on $\sqcap$, may exclude some simultaneous event occurrences.

    The lifting of composability relations distributes across the intersection.\footnote{The lifting does not distribute across the union, however.}
\begin{lemma}
    \label{lemma:intersection}
    For all composability relations $\kappa_1, \kappa_2$ and interfaces $E_1, E_2$:
\[
    [\kappa_1 \cap \kappa_2](E_1, E_2) = [\kappa_1](E_1, E_2) \cap [\kappa_2](E_1, E_2)
\]
\end{lemma}
\begin{proof}
    \begin{align*}
        [\kappa_1](E_1, E_2) \cap [\kappa_2](E_1, E_2) 
        &= \bigcup \{ \Rel \mid \Rel \subseteq \Phi_{\kappa_1}(E_1, E_2)(\Rel)\} \cap \bigcup \{ \Rel \mid \Rel \subseteq \Phi_{\kappa_2}(E_1, E_2)(\Rel)\} \\
        &= \bigcup \{ \Rel \mid \Rel \subseteq \Phi_{\kappa_1}(E_1, E_2)(\Rel) \mathit{\ and\ } \Rel \subseteq \Phi_{\kappa_2}(E_1, E_2)(\Rel)\} \\
        &= \bigcup \{ \Rel \mid \Rel \subseteq \Phi_{\kappa_1}(E_1, E_2)(\Rel) \cap \Phi_{\kappa_2}(E_1, E_2)(\Rel)\} \\
        &= \bigcup \{ \Rel \mid \Rel \subseteq \Phi_{\kappa_1 \cap \kappa_2}(E_1, E_2)(\Rel) \}\\
        &= [\kappa_1 \cap \kappa_2](E_1, E_2) \}
    \end{align*}
    since 
    \begin{align*}
    \Phi_{\kappa_1}(E_1, E_2)(\Rel) \cap \Phi_{\kappa_2}(E_1, E_2)(\Rel)
      &= \{(\tau_1, \tau_2) \mid  (\tau_1(0), \tau_2(0)) \in \kappa_1(E_1,E_2) \land  (\tau_1(0), \tau_2(0)) \in \kappa_2(E_1,E_2) \\
                                             &\hspace{4em}  (\pr_2(\tau_1)(0) =t_1 \land  \pr_2(\tau_2)(0) =t_2) \land \\
                                             &\hspace{4em} (t_1 < t_2 \land (\tau_1',\tau_2) \in \Rel \lor t_2 < t_1 \land (\tau_1,\tau_2') \in \Rel \lor  \\
                                             &\hspace{4em} \ t_2 = t_1 \land (\tau_1',\tau_2') \in \Rel)\} \\
                                             &= \{(\tau_1, \tau_2) \mid  (\tau_1(0), \tau_2(0)) \in \kappa_1(E_1,E_2) \cap \kappa_2(E_1,E_2) \\
                                             &\hspace{4em} (\pr_2(\tau_1)(0) =t_1 \land  \pr_2(\tau_2)(0) =t_2) \land \\
                                             &\hspace{4em}(t_1 < t_2 \land (\tau_1',\tau_2) \in \Rel \lor t_2 < t_1 \land (\tau_1,\tau_2') \in \Rel \lor  \\
                                             &\hspace{4em}\ t_2 = t_1 \land (\tau_1',\tau_2') \in \Rel)\} \\
                                             &= \Phi_{\kappa_1 \cap \kappa_2}(E_1, E_2)(\Rel)
    \end{align*}
\end{proof}

Similarly, we give a mechanism to lift a composition function on observables to a composition function on \tes s. 
Such lifting operation interleaves observations with different time stamps, and composes observations that occur at the same time.
\begin{definition}[Lifting - composition function]\label{def:lifting-composition}
    Let $+: \Po(\E) \times \Po(\E) \rightarrow \Po(\E)$ be a composition function on observables.
    The \emph{lifting} of $+$ to \tes s is $[+]: \TES{\E}{} \times \TES{\E}{} \rightarrow \TES{\E}{}$ such that, for $\sigma_i \in \TES{\E}{}$ where $\sigma_i(0) = (O_i, t_i)$ with $i \in \{1,2\}$:
    \[
    \sigma_1[+]\sigma_2 = 
    \begin{cases} 
        \langle \sigma_1(0)\rangle \cdot (\sigma_1'[+] \sigma_2)   &\mathit{if}\ t_1 < t_2 \\
        \langle \sigma_2(0)\rangle \cdot (\sigma_1[+] \sigma_2')     &\mathit{if}\ t_2 < t_1 \\
        \langle (O_1 + O_2, t_1)\rangle \cdot (\sigma_1'[+] \sigma_2') &  \mathit{otherwise}
    \end{cases} 
    \]
\end{definition}
Definition~\ref{def:lifting-composition} composes observations only if their time stamp is the same. Alternative definitions might consider time intervals instead of exact times.
\begin{example}[Intersection]
    For any two components $C_1 = (E_1, L_1)$ and $C_2 = (E_2, L_2)$, we define the intersection $C_1 \cap C_2$ to be the component $C_1 \times_{([\kappa_\sqcap^\ssync], [\cap])} C_2 = (E_1 \cup E_2,L)$
 where $\sqcap \subseteq E_1 \times E_2$ is such that $(O,O) \in \sqcap$ for all non-empty $O \subseteq E_1 \cup E_2$. 
 \label{ex:intersection}
    \hfill $\blacksquare$
\end{example}
\begin{example}[Join]
    For any two components $C_1 = (E_1, L_1)$ and $C_2 = (E_2, L_2)$, we define the join $C_1 \bowtie C_2$ to be the component $C_1 \times_{([\kappa_\sqcap^\ssync], [\cup])} C_2 = (E_1 \cup E_2,L)$
    where $\sqcap \subseteq E_1 \times E_2$ is such that $(O,O) \in \sqcap$ for all non-empty $O \subseteq E_1 \cap E_2$. 
 Note that the join of two components contains more behavior than the intersection of those two components: independent events (i.e., events not in $E_1 \cap E_2$) may occur freely in any observation.
 If $E_1 = E_2$, however, then $C_1 \bowtie C_2 = C_1 \cap C_2$.
 \label{ex:join}
    \hfill $\blacksquare$
\end{example}
\begin{lemma}
    \label{lemma:intersection-lift}
    Let $\kappa_1$ and $\kappa_2$ be two composability relations and $\times_{([\kappa_1 \cap \kappa_2],\oplus)}$ be a product on components.
    Then,
    \[
        C_1  \times_{([\kappa_1 \cap \kappa_2],\oplus)} C_2 = C_1  \times_{([\kappa_1] \cap [\kappa_2],\oplus)} C_2 =  (C_1  \times_{([\kappa_1 ],\oplus)} C_2) \cap  (C_1  \times_{([\kappa_2],\oplus)} C_2)
    \]
\end{lemma}
\begin{proof}
    Let $C_1  \times_{[\kappa_1 \cap \kappa_2],\oplus} C_2 = (E,L)$ and $(C_1  \times_{([\kappa_1 ],\oplus)} C_2) \cap  (C_1  \times_{([\kappa_2],\oplus)} C_2) = (E',L')$. We have $E = E_1 \cup E_2 = E'$. We show $L = L'$.
    \begin{align*}
        L&= \{\sigma_1 \oplus \sigma_2 \mid \sigma_1 \in L_1, \sigma_2 \in L_2, (\sigma_1, \sigma_2) \in [\kappa_1 \cap \kappa_2](E_1, E_2) \} \\
         &= \{\sigma_1 \oplus \sigma_2 \mid \sigma_1 \in L_1, \sigma_2 \in L_2, (\sigma_1, \sigma_2) \in [\kappa_1](E_1, E_2) \cap [\kappa_2](E_1, E_2) \} \\
         &= L'  
    \end{align*}
\end{proof}
\begin{example}\label{ex:mix}
    Composability relations as defined in Definition~\ref{def:sync} and Definition~\ref{def:excl} can be combined to form new relations, and therefore new products. 
The behavior of component $C_1 \times_{([\kappa^{\ssync}_{\sqcap} \cap \kappa^{\eexcl}_{\sqcap}]), \oplus)} C_2$ contains all \tes s that are in the behavior of both $C_1 \times_{([\kappa^{\ssync}_{\sqcap}], \oplus)} C_2$ and $C_1 \times_{([\kappa^{\eexcl}_{\sqcap}], \oplus)} C_2$, which excludes observations containing an occurrence of at least one event related by $\sqcap$.
    \hfill$\blacksquare$
\end{example}

\newcommand{\alt}{\mathit{alt}}
\newcommand{\intl}{\mathit{intl}}
\renewcommand{\excl}{\mathit{excl}}
    To show the expressiveness of our model, we consider its application to the Reo coordination language~\cite{AR03} in a series of example (Examples~\ref{ex:reo-comp},~\ref{ex:reo-prod}).
    Reo is a coordination language for components, and makes use of shared ports between components to synchronize their observables. 
    If two components share a port, then they ``agree'' on the value of the data that flow through that port.
    A component has a set of ports as interface, and defines a relation over its port values over time.
    Composition of two components is taking the conjunction of their relations.
\begin{example}
    [Reo components]
    \label{ex:reo-comp}
    Let $P$ be the universal set of port names in Reo, and $V(p)$ the set of values over port $p\in P$.
    Consider the set of events $\E = \{ (p, v) \mid p \in P,\ v \in V(p)\}$ that contains all pairs of a port name and a value, and $E_a\subseteq \E$ be the set of all events for port $a$ only, i.e., $E_a = \{(a,v) \mid v \in V(a)\}$.
    A port $a \in P$ corresponds to the component $C_a = (E_a, \TES{E_a}{})$.\\
    For any ports $a, b$ and function $f: V(a) \to V(b)$, we introduce a composability relation on observables $\sqcap_{(a,b,f)} \subseteq \Po(E_a) \times \Po(E_b)$ such that, for all $(O_1, O_2) \in \sqcap_{(a,b,f)}$, there exists $v \in V(a)$ with $(a,v) \in O_1$ if and only if $(b,f(v)) \in O_2$.
    Intuitively, such composability relation relate two observables $O_1$ and $O_2$ such that the value of $b$ in $O_2$ coincides with the image of the value of $a$ in $O_1$ under $f$.
    We omit $f$ when $f$ is the identity, and write $\sqcap_{(a,b)}$.
    We use $\cup$ as composition function on observables and write $\times_R$ as a shorthand for $\times_{(R,[\cup])}$. 

    The composability relation $\kappa^{intl}_{\sqcap}$ interleaves observations while restricting every occurrence of observable $O_1$ to be related to a later observable $O_2$, i.e., $((O_1, t_1), (O_2,t_2)) \in \kappa^{intl}_{\sqcap}$ if and only if 
    \[
        t_2 < t_1 \lor (t_1 < t_2\land (O_1, O_2) \in \sqcap)
    \]
    We define the alternating component: $M = (E_m, L_m)$ where $E_m = \{(m,0),(m,1)\}$ and $\sigma \in L_m\subseteq \TES{E_m}{}$ if and only if, for all $i \in \N$, $\sigma(2i) = (\{(m,0)\},t_i)$ and $\sigma(2i+1)= (\{(m,1)\},t_{i+1})$, which consists of a stream of alternating bits.
    Then, fixing the function $f_0$ such that $f_0(v)= 0$ for all $v \in V(a)$, the product $P_a \times_{\kappa^\ssync_{\sqcap_{(a,m,f_0)}}} M$ represents the component that synchronizes all values at port $a$ with the value $0$ at port $m$. Reciprocally, fixing the function $f_1$ such that $f_1(v) = 1$ for all $v \in V(b)$, the product $P_b \times_{\kappa^\ssync_{\sqcap_{(b,m,f_1)}}} M$ represents the component that synchronizes all values at port $b$ with the value $1$ at port $m$.
    \\
    We define the following Reo components:
    \begin{align*}
        \mathit{Sync}(a,b)          &= P_a \times_{[\kappa^{\ssync}_{\sqcap_{(a,b)}}]} P_b \\
        \mathit{Syncdrain}(a,b)     &= P_a \times_{[\kappa^{\ssync}_{\sqcap}]} P_b \textit{ with } \sqcap = \Po(E_a) \times \Po(E_b) \\
        \mathit{Fifo}(a,b,M)          &= (P_a \times_{[\kappa^{\ssync}_{\sqcap_{(a,m,f_0)}}]} M) \times_{[\kappa^{\intl}_{\sqcap_{(a,b)}}] \cup [\kappa^{\ssync}_{\sqcap_{(m,m)}}]} (P_b\times_{[\kappa^{\ssync}_{\sqcap_{(b,m,f_1)}}]} M)\\
        \mathit{Merger}(a,b,c)      &= (P_a \times_{[\kappa^{\excl}_{\sqcap_{(a,b)}}]} P_b) \times_{[\kappa^{\ssync}_{\sqcap_{(a,c)} \cup \sqcap_{(b,c)}}]} P_c 
    \end{align*}
    The $\mathit{Sync}(a,b)$ component is such that the data observed at port $a$ and $b$ are equal and synchronous, i.e., occurs at the same time. The $\mathit{Syncdrain}(a,b)$ component ensures that both the data of $a$ and $b$ are observed at the same time, but does not restric their data to be equal. The component $\mathit{Fifo(a,b,M)}$ synchronizes the observation of a data at $a$ with the change of the memory state $M$, and then outputs the same data at $b$. As defined here, the $\mathit{Fifo(a,b,M)}$ component is infinitely productive, i.e., always eventually has an input at $a$ and an output at $b$. The $\mathit{Merger(a,b,c)}$ component either synchronizes $a$ with $c$ or $b$ with $c$ but never all ports together.

    A strength of Reo is its compositional nature: protocols are built out of primitives. We use the join operation defined in Example~\ref{ex:join} (see Corollary~\ref{corollary:lifted-sync} for the proof of associativity and commutativity of $\bowtie$) to define two Reo components:
    \begin{align*}
        \mathit{Alternator}(a,b,c)  &= \mathit{Sync}(a,c_1) \bowtie \mathit{Fifo}(x,c_2) \bowtie \mathit{Syncdrain}(a,b)\bowtie \mathit{Sync}(b,x) \bowtie \mathit{Merger}(c_1,c_2,c)\\
        \mathit{Fifo_2}(a,b)  &= \mathit{Fifo}(a,x,M_1) \bowtie \mathit{Fifo}(x,b,M_2)
    \end{align*}
    \hfill $\blacksquare$
\end{example}

In order to state some sufficient conditions for $\kappa$ so that its lifted product satisfies the condition for associativity, we introduce the unique \emph{enumeration} of a triple of TESs. Intuitively, the enumeration is a stream that increments the index of each TES in order. 
Observe, from Definition~\ref{def:enumeration}, that for any triple $(\sigma_1, \sigma_2, \sigma_3)$, there exists a unique such enumeration.
\begin{definition}[Enumeration]
    \label{def:enumeration}
Let $(\sigma_1, \sigma_2, \sigma_3) \in \TES{E_1}{}\times\TES{E_2}{}\times\TES{E_3}{}$, an enumeration $\tau: \N \to (\N \times \N \times \N)$ is inductively defined by 
$\tau(0) = (0,0,0)$ and, given $\tau(n) = (i,j,k)$, 
\[\tau(n+1) = 
\begin{cases} 
    (i+1,j,k) & \textit{ if } \pr_2(\sigma_1)(i) < \pr_2(\sigma_2)(j) \land \pr_2(\sigma_1)(i) < \pr_2(\sigma_3)(k) \\
    (i,j+1,k) & \textit{ if } \pr_2(\sigma_2)(j) < \pr_2(\sigma_1)(i) \land \pr_2(\sigma_2)(j) < \pr_2(\sigma_3)(k) \\
    (i,j,k+1) & \textit{ if } \pr_2(\sigma_3)(i) < \pr_2(\sigma_2)(j) \land \pr_2(\sigma_3)(k) < \pr_2(\sigma_1)(i) \\
    (i+1,j+1,k) & \textit{ if } \pr_2(\sigma_1)(i) = \pr_2(\sigma_2)(j) < \pr_2(\sigma_3)(k) \\
    (i+1,j,k+1) & \textit{ if } \pr_2(\sigma_1)(i) = \pr_2(\sigma_3)(k) < \pr_2(\sigma_2)(j) \\
    (i,j+1,k+1) & \textit{ if } \pr_2(\sigma_2)(j) = \pr_2(\sigma_3)(k) < \pr_2(\sigma_1)(i) \\
    (i+1,j+1,k+1) & \textit{ if } \pr_2(\sigma_1)(i) = \pr_2(\sigma_2)(j) = \pr_2(\sigma_3)(k) 
\end{cases}\]
\end{definition}

Besides the product instances detailed in Definitions~\ref{def:cross},~\ref{def:sync}, ~\ref{def:excl}\red{, and~\ref{ex:mix}}, the definition of composability relation or composition function as the lift of some composability relation on observations or function on observables allows sufficient conditions for Lemma~\ref{lemma:prod} to hold. 

\begin{lemma}
    \label{lemma:lifting-properties}
    Let $+$ be a composition function on observables and let $\kappa$ be a composability relation on observations.
    Then,
    \begin{itemize}
    \item $\times_{([\kappa],[+])}$ is \emph{commutative} if $\kappa$ is symmetric and $+$ is commutative; 
    \item $\times_{([\kappa],[+])}$ is \emph{associative} if $+$ is associative and, for all $E_1, E_2, E_3$ and for all $(\sigma_1, \sigma_2, \sigma_3) \in \TES{E_1}{} \times \TES{E_2}{} \times \TES{E_3}{}$, let $\tau$ be an enumeration of $(\sigma_1, \sigma_2, \sigma_3)$, then, for all $n \in \N$, letting $\tau(n) = (i,j,k)$,
    \begin{align*}
            (\sigma_1(i),\sigma_2(j)) \in \kappa(E_1,E_2)  \land ((\sigma^{(i)}_1[+] \sigma^{(j)}_2)(0),\sigma_3(k)) \in \kappa (E_1 \cup E_2,E_3)
    \end{align*}
            if and only if,
for all $n \in \N$, letting $\tau(n) = (i,j,k)$,
    \begin{align*}
    (\sigma_2(j),\sigma_3(k)) \in \kappa(E_2,E_3)  \land (\sigma_1(i),(\sigma_2^{(j)}[+]\sigma_3^{(k)})(0)) \in \kappa (E_1,E_2\cup E_3)
        \end{align*}
    \item $\times_{([\kappa],[+])}$ is \emph{idempotent} if $+$ is idempotent and, for all $E \subseteq \E$ we have $((O_1,t_1), (O_2, t_2)) \in \kappa(E, E) \implies (O_1, t_1) = (O_2, t_2)$.
    \end{itemize}
\end{lemma}
\begin{proof}
    \underline{Commutativity.} From Lemma~\ref{lemma:lifting}, if $\kappa$ is symmetric, then its lifting $[\kappa]$ is also symmetric. Therefore, it is sufficient for $\kappa$ to be symmetric and for $+$ to be commutative in order for $[\kappa]$ to be symmetric and $[+]$ to be commutative, and therefore $\times_{([\kappa],[+])}$ to be commutative.

    \hspace{-0.3em}\underline{Associativity.} \hspace{-0.2em}
    A sufficient condition for the product $\times_{([\kappa], [+])}$ to be associative is that $+$ is associative and  
    for every $\sigma_i \in \TES{E_i}{}$ for $i \in \{1,2,3\}$:
    \[
    \begin{array}{rl}
        P_1 :=&  (\sigma_1,\sigma_2) \in [\kappa](E_1,E_2)  \land (\sigma_1[+] \sigma_2,\sigma_3) \in [\kappa] (E_1 \cup E_2,E_3) \iff\\
             &(\sigma_2,\sigma_3) \in [\kappa] (E_2,E_3)  \land (\sigma_1, \sigma_2 [+] \sigma_3) \in [\kappa] (E_1,E_2\cup E_3) 
    \end{array}
    \]
    We introduce the function
    \[
    \begin{array}{rl}
        \Psi^1_{\kappa}(E_1, E_2, E_3)(\Rel) = \{ &(\sigma_1, \sigma_2, \sigma_3) \mid  (\sigma_1(0),\sigma_2(0)) \in \kappa(E_1,E_2)  \land \\ 
                                                            &((\sigma_1[+] \sigma_2)(0),\sigma_3(0)) \in \kappa (E_1 \cup E_2,E_3) \land   \\
                                                            & \pr_2(\sigma_1)(0) =t_1 \land  \pr_2(\sigma_2)(0) =t_2 \land  \pr_2(\sigma_3)(0) =t_3 \land \\
                                                            &( t_1 < t_2 \land t_1 < t_3 \land (\sigma_1',\sigma_2, \sigma_3) \in \Rel   \ \lor \\
                                                            &\ t_2 < t_1 \land t_2 < t_3 \land (\sigma_1,\sigma_2', \sigma_3) \in \Rel  \ \lor \\
                                                            &\ t_3 < t_2 \land t_3 < t_1 \land (\sigma_1,\sigma_2, \sigma_3') \in \Rel  \ \lor \\
                                                            &\ t_1 = t_2 \land t_1 < t_3 \land (\sigma_1',\sigma_2', \sigma_3) \in \Rel \ \lor \\
                                                            &\ t_2 = t_3 \land t_2 < t_1 \land (\sigma_1,\sigma_2', \sigma_3') \in \Rel \ \lor \\
                                                            &\ t_1 = t_3 \land t_1 < t_2 \land (\sigma_1',\sigma_2, \sigma_3') \in \Rel \ \lor \\
                                                            &\ t_1 = t_3 \land t_1 = t_2 \land (\sigma_1',\sigma_2', \sigma_3') \in \Rel )\}\\
    \end{array}
    \]
    and 
    \[
    \begin{array}{rl}
        \Psi^2_{\kappa}(E_1, E_2, E_3)(\Rel) = \{ &(\sigma_1, \sigma_2, \sigma_3) \mid  (\sigma_2(0),\sigma_3(0)) \in \kappa(E_2,E_3)  \land \\ 
                                                              &(\sigma_1(0),(\sigma_2[+]\sigma_3)(0)) \in \kappa (E_1,E_2\cup E_3) \land   \\
                                                            & \pr_2(\sigma_1)(0) =t_1 \land  \pr_2(\sigma_2)(0) =t_2 \land  \pr_2(\sigma_3)(0) =t_3 \land \\
                                                            &( t_1 < t_2 \land t_1 < t_3 \land (\sigma_1',\sigma_2, \sigma_3) \in \Rel   \ \lor \\
                                                            &\ t_2 < t_1 \land t_2 < t_3 \land (\sigma_1,\sigma_2', \sigma_3) \in \Rel  \ \lor \\
                                                            &\ t_3 < t_2 \land t_3 < t_1 \land (\sigma_1,\sigma_2, \sigma_3') \in \Rel  \ \lor \\
                                                            &\ t_1 = t_2 \land t_1 < t_3 \land (\sigma_1',\sigma_2', \sigma_3) \in \Rel \ \lor \\
                                                            &\ t_2 = t_3 \land t_2 < t_1 \land (\sigma_1,\sigma_2', \sigma_3') \in \Rel \ \lor \\
                                                            &\ t_1 = t_3 \land t_1 < t_2 \land (\sigma_1',\sigma_2, \sigma_3') \in \Rel \ \lor \\
                                                            &\ t_1 = t_3 \land t_1 = t_2 \land (\sigma_1',\sigma_2', \sigma_3') \in \Rel )\}\\
    \end{array}
    \]
    Showing that $P_1$ holds is equivalent to showing that $\Rel \subseteq \Psi^1_{\kappa}(E_1, E_2, E_3)(\Rel)$ if and only if $\Rel \subseteq \Psi^2_{\kappa}(E_1, E_2, E_3)(\Rel)$ for any $\mathcal{R} \subseteq \TES{E_1}{} \times \TES{E_2}{} \times \TES{E_3}{}$.

    First, we observe that, given $\mathcal{R} \subseteq \Psi^1_{\kappa}(E_1, E_2, E_3)(\Rel)$, for any $(\sigma_1, \sigma_2, \sigma_3) \in \mathcal{R}$ and its enumeration $\tau$, for all $n \in \N$ letting $\tau(n) = (i,j,k)$, we have: 
    \begin{align*}
            (\sigma_1(i),\sigma_2(j)) \in \kappa(E_1,E_2)  \land ((\sigma^{(i)}_1[+] \sigma^{(j)}_2)(0),\sigma_3(k)) \in \kappa (E_1 \cup E_2,E_3)
    \end{align*}
    Similarly, if the same $\Rel$ is such that $\mathcal{R} \subseteq \Psi^2_{\kappa}(E_1, E_2, E_3)(\Rel)$, then, for all $n$ there exist $i,j,k$ such that $\tau(n) = (i,j,k)$ and 
    \begin{align*}
    (\sigma_2(j),\sigma_3(k)) \in \kappa(E_2,E_3)  \land (\sigma_1(i),(\sigma_2^{(j)}[+]\sigma_3^{(k)})(0)) \in \kappa (E_1,E_2\cup E_3)
    \end{align*}
    Thus, to show that $P_1$ holds, it is sufficient to prove that, for all $n \in \N$ there exists $(i,j,k)$ with  $\tau(n) = (i,j,k)$ and
    \begin{align*}
            (\sigma_1(i),\sigma_2(j)) \in \kappa(E_1,E_2)  \land ((\sigma^{(i)}_1[+] \sigma^{(j)}_2)(0),\sigma_3(k)) \in \kappa (E_1 \cup E_2,E_3)
    \end{align*}
    if and only if,
for all $n \in \N$ there exists $(i,j,k)$ with  $\tau(n) = (i,j,k)$ and
    \begin{align*}
    (\sigma_2(j),\sigma_3(k)) \in \kappa(E_2,E_3)  \land (\sigma_1(i),(\sigma_2^{(j)}[+]\sigma_3^{(k)})(0)) \in \kappa (E_1,E_2\cup E_3)
    \end{align*}
    The above equivalence holds for any $(\sigma_1, \sigma_2, \sigma_3) \in \TES{E_1}{} \times\TES{E_2}{} \times\TES{E_3}{}$ and $\tau$ the enumeration of $(\sigma_1, \sigma_2, \sigma_3)$.

    Finally, we prove that if $+$ is associative, then $[+]$ is associative.
    Let $\sigma_i \in L_i$ and we write $\sigma_i(0) = (O_i, t_i)$ for $i \in \{1,2,3\}$, then:
    \[
    \sigma_1 [+](\sigma_2 [+] \sigma_3) = 
    \begin{cases}
        \langle(O_1,t_1)\rangle\cdot (\sigma_1'[+](\sigma_2 [+]\sigma_3)                    &\hspace{-0.4em}\text{if } t_1 < t_2 \land t_1 < t_3 \\
        \langle(O_2,t_2)\rangle\cdot (\sigma_1 [+](\sigma_2'[+]\sigma_3)                    &\hspace{-0.4em}\text{if } t_2 < t_1 \land t_2 < t_3 \\
        \langle(O_3,t_3)\rangle\cdot (\sigma_1 [+](\sigma_2 [+]\sigma_3')                   &\hspace{-0.4em}\text{if } t_3 < t_2 \land t_3 < t_1 \\
        \langle(O_1+O_2,t_1)\rangle\cdot (\sigma_1'[+](\sigma_2'[+]\sigma_3)              &\hspace{-0.4em}\text{if } t_1 = t_2 \land t_1 < t_3 \\
        \langle(O_2+O_3,t_2)\rangle\cdot (\sigma_1 [+](\sigma_2'[+]\sigma_3')             &\hspace{-0.4em}\text{if } t_2 = t_3 \land t_2 < t_1 \\
        \langle(O_1+ O_3,t_1)\rangle\cdot (\sigma_1'[+](\sigma_2 [+]\sigma_3')            &\hspace{-0.4em}\text{if } t_1 = t_3 \land t_1 < t_2 \\
        \langle(O_1+(O_2 + O_3),t_1)\rangle\cdot (\sigma_1'[+](\sigma_2'[+]\sigma_3')   &\hspace{-0.4em}\text{if } t_1 = t_3 \land t_1 = t_2 \\
        \end{cases}
    \]

    The only case that differs from $(\sigma_1 [+]\sigma_2) [+] \sigma_3$ is when $t_1 = t_3 = t_2$, which gives 
    $((O_1+O_2) + O_3,t_1)$.
    Thus, if $((O_1+O_2) + O_3,t_1) = (O_1+(O_2 + O_3),t_1)$ for every $O_i \in \Po(E_i)$ with $i \in \{1,2,3\}$, then 
    $\sigma_1 [+] (\sigma_2 [+] \sigma_3) = \sigma_1 [+] (\sigma_2 [+] \sigma_3)$ 
    for every $\sigma_i \in L_i$ with $i \in \{1,2,3\}$.

    \underline{Idempotency.}
    If $+$ is idempotent, then the lifting $[+]$ is also idempotent.
    We consider $+$ to be idempotent.
    We show that, for all $E \subseteq \E$ and $o_1, o_2 \in \Po(E) \times \Rp$ we have $(o_1, o_2) \in \kappa(E,E) \implies o_1 = o_2$, then for all $\sigma, \tau \in \TES{E}{}$, $(\sigma, \tau) \in [\kappa](E,E) \implies \sigma = \tau$, which is a sufficient condition for $\times_{([\kappa],[+])}$ to be idempotent.

    By definition $[\kappa](E,E)$ is the greatest fixed point of the function:
    \[
    \begin{array}{rll}
        \Phi_{\kappa}(E,E)(\Rel) = \{(\tau_1, \tau_2) \mid & (\tau_1(0), \tau_2(0)) \in \kappa_1(E,E) \land   \\
          & (\pr_2(\tau_1)(0) =t_1 \land  \pr_2(\tau_2)(0) =t_2) \land \\
          &(t_1 < t_2 \land (\tau_1',\tau_2) \in \Rel \lor t_2 < t_1 \land (\tau_1,\tau_2') \in \Rel \\
          &\ t_2 = t_1 \land (\tau_1',\tau_2') \in \Rel)\}\\
            \subseteq \{(\tau_1, \tau_2) \mid & \tau_1(0) = \tau_2(0) \land  (\tau_1',\tau_2') \in \Rel \}
    \end{array}
\]
    Therefore, we conclude that $[\kappa](E,E) \subseteq \{(\sigma, \sigma) \mid \sigma \in \TES{E}{}\}$.
\end{proof}

The result of Lemma~\ref{lemma:lifting-properties} can be expressed in terms of conditions on composability relations on observables. In Property~\ref{prop:sqcap}, we introduce some sufficient conditions for $\times_{([\kappa_\sqcap^\ssync], [\cup])}$ to be associative (i.e., to satisfy conditions in Lemma~\ref{lemma:lifting-properties}).
\begin{property}
    We list a series of properties on a composability relation on observable $\sqcap$ where, for $O_1, O_2, O_3 \subseteq \E$:
    \begin{enumerate}
        \item\label{sq1} 
            if $\ind_\sqcap(O_1, O_2)$ and $\ind_\sqcap(O_1, O_3)$ then $\ind_\sqcap(O_1, O_2 \cup O_3)$; and 
            if $\ind_\sqcap(O_2, O_3)$ and $\ind_\sqcap(O_1, O_3)$ then $\ind_\sqcap(O_1\cup O_2, O_3)$
        \item\label{sq2} if $(O_1, O_2) \in \sqcap$ and $(O_2, O_3) \in \sqcap$, then $(O_1, O_3) \in \sqcap$; 
        \item\label{sq3} if $(O_1, O_2) \in \sqcap$ and $(O_1, O_3) \in \sqcap$, then $(O_2, O_3) \in \sqcap$ and $(O_3, O_2) \in \sqcap$;
        \item\label{sq4} if $(O_1, O_3) \in \sqcap$ and $(O_2, O_3) \in \sqcap$, then $(O_1, O_2) \in \sqcap$ and $(O_2, O_1) \in \sqcap$; 
        \item\label{sq5} 
            if $(O_1, O_2) \in \sqcap$ and $(O_1',O_2') \in \sqcap$, then $(O_1\cup O_1', O_2 \cup O_2') \in \sqcap$.
    \end{enumerate}
    \label{prop:sqcap}
\end{property}
\begin{corollary}
    \label{corollary:lifted-sync}
    Let $\sqcap$ be a composability relation on observables and $\kappa^\ssync_\sqcap$ as defined in Example~\ref{def:sync}.
    Then:
    \begin{itemize}
        \item if $\sqcap$ is symmetric, then the product $\times_{([\kappa^\ssync_\sqcap],[\cup])}$ is commutative;
        \item if $\sqcap$ satisfies Properties~\ref{prop:sqcap}, then the product $\times_{([\kappa^\ssync_\sqcap],[\cup])}$ is associative;
        \item if $\sqcap$ is co-reflexive, i.e., $(O_1,O_2) \in \sqcap$ implies $O_1 = O_2$,  then the product $\times_{([\kappa^\ssync_\sqcap],[\cup])}$ is idempotent.
    \end{itemize}
\end{corollary}
\begin{proof}
    Let $E_1, E_2, E_3 \subseteq \E$, and $\sqcap \subseteq \Po(E_1\cup E_2 \cup E_3) \times \Po(E_1\cup E_2 \cup E_3)$.
    We use $\kappa^\ssync_\sqcap(E_1, E_2)$, the composability relation in Definition~\ref{def:sync}:
    \begin{itemize}
        \item for all $(O_1, O_2) \in \Po(E_1)\times \Po(E_2)$ such that $(O_1, O_2) \in \sqcap$ and for all $(O_1', O_2') \in \Po(E_1)\times \Po(E_2)$ such that $\ind_\sqcap(O_1', E_2)$ and $\ind_\sqcap(E_1, O_2')$ then, for all time stamps $t$, $((O_1\cup O_1', t), (O_2 \cup O_2',t)) \in \kappa_\sqcap^{\ssync}(E_1, E_2)$.
        \item  if $\ind_\sqcap(O_1, E_2)$ then for all $O_2 \subseteq E_2$ and $t_1 \leq t_2$, $((O_1,t_1), (O_2,t_2)) \in \kappa^\ssync_\sqcap(E_1,E_2)$. Reciprocally, if $\ind_\sqcap(E_1, O_2)$ then for all $O_1 \subseteq E_1$ and $t_2 \leq t_1$, $((O_1,t_1), (O_2,t_2)) \in \kappa^\ssync_\sqcap(E_1,E_2)$;
    \end{itemize}
    where $\ind_\sqcap(X,Y) = \forall x \subseteq X.\ \forall y \subseteq Y.\ (x,y)\not \in \sqcap$.

    Let $(\sigma_1, \sigma_2, \sigma_3) \in \TES{E_1}{} \times \TES{E_2}{} \times \TES{E_3}{}$ and let $\tau$ be an enumeration of $(\sigma_1, \sigma_2, \sigma_3)$. 
    We show that, if, for all $n \in \N$, letting $\tau(n) = (i,j,k)$, the following holds 
    \begin{align*}
            (\sigma_1(i),\sigma_2(j)) \in \kappa^\ssync_\sqcap(E_1,E_2)  \land ((\sigma^{(i)}_1[\cup] \sigma^{(j)}_2)(0),\sigma_3(k)) \in \kappa^\ssync_\sqcap (E_1 \cup E_2,E_3)
    \end{align*}
          then, for all $n \in \N$, letting $\tau(n) = (i,j,k)$, the following holds 
    \begin{align*}
    (\sigma_2(j),\sigma_3(k)) \in \kappa^\ssync_\sqcap(E_2,E_3)  \land (\sigma_1(i),(\sigma_2^{(j)}[\cup]\sigma_3^{(k)})(0)) \in \kappa^\ssync_\sqcap (E_1,E_2\cup E_3)
\end{align*}  (the implication in the other direction is similar).
We proceed by induction on $n$ (i.e., the elements of the enumeration) in the right hand side of the implication.\\
\underline{Base case.} Let $n = 0$. Then, $\tau(n) = \tau(0) =  (0, 0, 0)$.
We show that 
\[
    (\sigma_2(0),\sigma_3(0)) \in \kappa^\ssync_\sqcap(E_2,E_3)  \land (\sigma_1(0),(\sigma_2[\cup]\sigma_3)(0)) \in \kappa^\ssync_\sqcap (E_1,E_2\cup E_3)
\]
We know that, for all $n$, letting $\tau(n) = (i,j,k)$, we have 
    \begin{align*}
            (\sigma_1(i),\sigma_2(j)) \in \kappa^\ssync_\sqcap(E_1,E_2)  \land ((\sigma^{(i)}_1[\cup] \sigma^{(j)}_2)(0),\sigma_3(k)) \in \kappa^\ssync_\sqcap (E_1 \cup E_2,E_3)
    \end{align*}
    Thus, for $n = 0$, we know that
    \begin{align*}
            (\sigma_1(0),\sigma_2(0)) \in \kappa^\ssync_\sqcap(E_1,E_2)  \land ((\sigma_1[\cup] \sigma_2)(0),\sigma_3(0)) \in \kappa^\ssync_\sqcap (E_1 \cup E_2,E_3)
    \end{align*}
    Let $\sigma_1(0) = (O_1,t_1)$, $\sigma_2(0) = (O_2, t_2)$, and $\sigma_3(0) = (O_3,t_3)$, we split cases on time stamp values and show the implication for each of those cases.
\begin{itemize}
    \item Let $t_1 < t_2 \land t_1 < t_3$. 
        Then 
    \begin{align*}
            (\sigma_1(0),\sigma_2(0)) \in \kappa^\ssync_\sqcap(E_1,E_2)  \land (\sigma_1(0),\sigma_3(0)) \in \kappa^\ssync_\sqcap (E_1 \cup E_2,E_3)
    \end{align*}
    which implies that 
    $(\sigma_1(0),(\sigma_2 [\cup] \sigma_3)(0)) \in \kappa^\ssync_\sqcap (E_1,E_2\cup E_3)$ by 
    the fact that $\ind_\sqcap(O_1, E_2)$ and $\ind_\sqcap(O_1, E_3)$ together imply that $\ind_\sqcap(O_1, E_2 \cup E_3)$.
    Then, one can show that, given the property of the enumeration, there exists an $i'$ such that $\tau(i') = (i',0,0)$ with $t_2 \leq \pr_2(\sigma_1)(i')$ or $t_3 \leq \pr_2(\sigma_1)(i')$. We discuss those cases below to prove that $(\sigma_2(0),\sigma_3(0)) \in \kappa^\ssync_\sqcap(E_2,E_3)$.
    \item Let $t_2 < t_1 \land t_2 < t_3$. Then,
    \begin{align*}
            (\sigma_1(0),\sigma_2(0)) \in \kappa^\ssync_\sqcap(E_1,E_2)  \land (\sigma_2(0),\sigma_3(0)) \in \kappa^\ssync_\sqcap (E_1 \cup E_2,E_3)
    \end{align*}
    The definition of $\kappa_\sqcap^\ssync$ implies that, if $((O,t),(O',t')) \in \kappa^\ssync(E,E')$ and $t<t'$ then $((O,t),(O',t')) \in \kappa^\ssync_\sqcap(E'',E')$ for arbitrary $E''$ with $O \subseteq E''$. Symmetrically, if $t'<t$, then $((O,t),(O',t')) \in \kappa^\ssync_\sqcap(E,E'')$ for arbitrary $E''$ with $O' \subseteq E''$.
    Thus, 
    \[    (\sigma_2(0),\sigma_3(0)) \in \kappa^\ssync_\sqcap(E_2,E_3)  \land (\sigma_1(0),\sigma_2(0)) \in \kappa^\ssync_\sqcap (E_1,E_2\cup E_3)\]
    \item Let $t_3 < t_2 \land t_3 < t_1$. Then,
    \begin{align*}
            (\sigma_1(0),\sigma_2(0)) \in \kappa^\ssync_\sqcap(E_1,E_2)  \land ((\sigma_1[\cup] \sigma_2)(0),\sigma_3(0)) \in \kappa^\ssync_\sqcap (E_1 \cup E_2,E_3)
    \end{align*}
    By definition of $\ind_\sqcap$, if $\ind_\sqcap(E_1\cup E_2, O_3)$, then $\ind_\sqcap(E_2, O_3)$ and $\ind_\sqcap(E_1, O_3)$. Thus,
\[
    (\sigma_2(0),\sigma_3(0)) \in \kappa^\ssync_\sqcap(E_2,E_3)  \land (\sigma_1(0),\sigma_3(0)) \in \kappa^\ssync_\sqcap (E_1,E_2\cup E_3)
\]

    \item Let $t_1 = t_2 \land t_1 < t_3$.
        Then,
    \begin{align*}
            (\sigma_1(0),\sigma_2(0)) \in \kappa^\ssync_\sqcap(E_1,E_2)  \land ((\sigma_1[\cup] \sigma_2)(0),\sigma_3(0)) \in \kappa^\ssync_\sqcap (E_1 \cup E_2,E_3)
    \end{align*}
    By definition of $\kappa_\sqcap^\ssync$, there exists $O_1' \subseteq O_1$ and $O_2' \subseteq O_2$ such that $(O_1',O_2') \in \sqcap$ with $\ind_\sqcap(O_1\setminus O_1', E_2)$ and $\ind_\sqcap(E_1, O_2\setminus O_2')$.
    Moreover, $\ind_\sqcap(O_1 \cup O_2, E_3)$ 
    implies that  $\ind_\sqcap(O_2, E_3)$ and $\ind_\sqcap(O_1, E_3)$ using 
    the definition of $\ind_\sqcap$.
    Thus, 
\[    (\sigma_2(0),\sigma_3(0)) \in \kappa^\ssync_\sqcap(E_2,E_3)  \land (\sigma_1(0),\sigma_2(0)) \in \kappa^\ssync_\sqcap (E_1,E_2\cup E_3)
\]
    \item Let $t_2 = t_3 \land t_2 < t_1$.
        Then,
    \begin{align*}
            (\sigma_1(0),\sigma_2(0)) \in \kappa^\ssync_\sqcap(E_1,E_2)  \land (\sigma_2(0),\sigma_3(0)) \in \kappa^\ssync_\sqcap (E_1 \cup E_2,E_3)
    \end{align*}
    Similar arguments as for the case where $t_1 = t_2 \land t_1 < t_3$ help to conclude that
\[    (\sigma_2(0),\sigma_3(0)) \in \kappa^\ssync_\sqcap(E_2,E_3)  \land (\sigma_1(0),(\sigma_2[\cup]\sigma_3)(0)) \in \kappa^\ssync_\sqcap (E_1,E_2\cup E_3)
\]
    \item Let $t_1 = t_3 \land t_1 < t_2$.
        Then,
    \begin{align*}
            (\sigma_1(0),\sigma_2(0)) \in \kappa^\ssync_\sqcap(E_1,E_2)  \land (\sigma_1(0),\sigma_3(0)) \in \kappa^\ssync_\sqcap (E_1 \cup E_2,E_3)
    \end{align*}
    By definition of $\kappa_\sqcap^\ssync$, there exist $O_1' \subseteq O_1$ and $O_3' \subseteq O_3$ 
                                such that $(O_1', O_3') \in \sqcap$ with $\ind_\sqcap(O_1\setminus O_1',E_3)$ and $\ind_\sqcap(E_1 \cup E_2,O_3\setminus O_3')$ and
                        $\ind_\sqcap(O_1, E_2)$.
                        Moreover, Property 1 item~\ref{sq1} together with 
                        $\ind_\sqcap(O_1, E_2)$ and $\ind_\sqcap(O_1\setminus O_1', E_3)$, we can conclude that
                        $\ind_\sqcap(O_1\setminus O_1',E_2 \cup E_3)$.
                        Also, since $\ind_\sqcap(O_1, E_2)$ and $(O_1', O_3') \in \sqcap$, it must be that $\ind_\sqcap(E_2, O_3')$ otherwise, there would be an $O_2' \subseteq E_2$ with $(O_2', O_3') \in \sqcap$ and therefore $(O_1', O_2') \in \sqcap$ by Property $1$ item~\ref{sq2}. 
                        Thus, $\ind_\sqcap(E_2, O_3)$, and
\[(\sigma_2(0),\sigma_3(0)) \in \kappa^\ssync_\sqcap(E_2,E_3)  \land (\sigma_1(0),\sigma_3(0)) \in \kappa^\ssync_\sqcap (E_1,E_2\cup E_3)\]
    \item Let $t_1 = t_3 \land t_1 = t_2$. Then,
    \begin{align*}
            (\sigma_1(0),\sigma_2(0)) \in \kappa^\ssync_\sqcap(E_1,E_2)  \land ((\sigma_1[\cup] \sigma_2)(0),\sigma_3(0)) \in \kappa^\ssync_\sqcap (E_1 \cup E_2,E_3)
    \end{align*}
    By definition of $\kappa_\sqcap^\ssync$,  there exist $O_1' \subseteq O_1$, $O_2' \subseteq O_2$
                                such that $(O_1', O_2') \in \sqcap$, $\ind_\sqcap(O_1\setminus O_1',E_2)$, and $\ind_\sqcap(E_1,O_2\setminus O_2')$; and 
                                there exist $O_1'' \subseteq O_1$, $O_2'' \subseteq O_2$, $O_3' \subseteq O_3$
                                such that $(O_1''\cup O_2'', O_3') \in \sqcap$ and 
                                $\ind_\sqcap(E_1,O_3\setminus O_3')$; and 
                                $\ind_\sqcap((O_1\cup O_2) \setminus (O_1''\cup O_2''), E_3)$.
                                If $O_2'' = \emptyset$, then $\ind_\sqcap(O_2, E_3)$ and $(\sigma_2(0), \sigma_3(0)) \in \kappa_\sqcap^\ssync(E_2, E_3)$; and if $O_2'' \not = \emptyset$, then, without loss of generality, we can assume that $\neg \ind_\sqcap(O_2'', E_3)$ and that there exists $O_2''' \subseteq O_2''$ and $O_3'' \subseteq O_3'$ such that $(O_2''', O_3'') \in \sqcap$ and $\ind_\sqcap(O_2\setminus O_2''', E_3)$ and $\ind_\sqcap(E_2, O_3\setminus O_3'')$.
                    Thus, $(\sigma_2(0),\sigma_3(0)) \in \kappa^\ssync_\sqcap(E_2,E_3)$.\\
                    Similarly, there exists $O_1''' \subseteq O_1''$ and $O_3''' \subseteq O_3'$ such that $\ind_\sqcap(O_1\setminus O_1''', E_3)$, $\ind_\sqcap(E_1, O_3\setminus O_3''')$, and $(O_1''', O_3''') \in \sqcap$, and therefore $(O_1''' \cup O_1', O_3''' \cup O_2') \in \sqcap$ by Property~\ref{prop:sqcap} item~\ref{sq5}.
                    Using Property 1 item~\ref{sq1} with $\ind_\sqcap(E_1,(O_2\setminus O_2') \cup (O_3\setminus O_3''))$, we get that $\ind_\sqcap(E_1,(O_2 \cup O_3)\setminus (O_2' \cup O_3'''))$ that and together with $\ind_\sqcap(O_1\setminus (O_1'''\cup O_1'),E_2\cup E_3)$, leads to the conclusion that $(\sigma_1(0),(\sigma_2[\cup]\sigma_3)(0)) \in \kappa^\ssync_\sqcap (E_1,E_2\cup E_3)$.
\end{itemize}
We showed the base case, for $n=0$.\\
\underline{Inductive argument.} Let $n$ be arbitrary, and $\tau(n) = (i,j,k)$. 
We show that the result holds for $n+1$ and $\tau(n+1) = (i',j',k')$.
The inductive argument relies on the fact that one may apply the initialization arguments for the triple $(\sigma_1^{(i)}, \sigma_2^{(j)}, \sigma_3^{(k)})$ with enumeration $\tau^{(n)}$. Based on the above arguments, one can therefore prove that the implication holds for $(i',j',k')$, and conclude that it holds for every $n$, which proves the statement.

\noindent
Reflexivity. If $\sqcap$ is co-reflexive, then $(O_1,O_2) \in \sqcap$ implies $O_1 = O_2$.
Suppose any two observations $(O_1, t_1), (O_2,t_2)$ with $O_1, O_2 \subseteq E$.
Due to the reflexivity of $\sqcap$, $(O_1,O_1) \in \sqcap$ and $(O_2,O_2) \in \sqcap$. 
It must be that $t_1 = t_2$.
Moreover, there is no $O_1',O_2' \subseteq E$ such that $(O_1, O_2') \in \sqcap$ or $(O_1', O_2) \in \sqcap$.
Thus, we can conclude that
$((O_1, t_1), (O_2, t_2)) \in \kappa^\ssync_\sqcap(E,E)$ implies $O_1 = O_2$ and $t_1 =t_2$. We therefore conclude that the product $\times_{([\kappa^\ssync_\sqcap],[\cup])}$ is idempotent.
\end{proof}
\begin{example}[Properties of join]
    \label{ex:reo-prod}
    Let $\sqcap = \{(O,O) \mid O \subseteq \E\}$ with $\E$ the universal set of events, then $\sqcap$ is co-reflexive, symmetric, and satisfies each item of Property~\ref{prop:sqcap}.
    As a consequence, the \emph{join} product $\bowtie$ defined in example~\ref{ex:join} is commutative, associative, and idempotent.
    Note that the join product is the product operator for Reo components.
\end{example}

We give in Lemma~\ref{lemma:distributivity} some conditions for two products to distribute, and in Lemma~\ref{lemma:extension} some conditions to extend the underlying relation on observables for a synchronous composability relation.
\begin{lemma}
    \label{lemma:distributivity}
Let $C_1$, $C_2$, and $C_3$ be three components, and let $\kappa_1$ and $\kappa_2$ be two composability relations on observables such that 
for all $\sigma_1, \sigma_2, \sigma_3 \in L_1 \times L_2 \times L_3$:
\begin{itemize}
    \item $(\sigma_1, \sigma_2[\cup] \sigma_3) \in [\kappa_1]$ if and only if $(\sigma_1, \sigma_2) \in [\kappa_1]$ and $(\sigma_1, \sigma_3) \in [\kappa_1]$, and
    \item for all $\tau_1 \in L_1$,  $(\tau_1[\cup]\sigma_2, \sigma_1[\cup]\sigma_3) \in [\kappa_2]$ if and only if $(\sigma_2, \sigma_3) \in [\kappa_2]$ and $\sigma_1 = \tau_1$.
\end{itemize}
Then, 
\[ 
     C_1 \times_{[\kappa_1]} (C_2 \times_{[\kappa_2]} C_3) = 
    (C_1 \times_{[\kappa_1]} C_2) \times_{[\kappa_2]} (C_1 \times_{[\kappa_1]}C_3)
\]
\end{lemma}
\begin{proof}
    Let $L$ be the behavior of component $(C_1 \times_{[\kappa_1]} C_2) \times_{[\kappa_2]} (C_1 \times_{[\kappa_1]}C_3)$, $L'$ be the behavior of $C_1 \times_{[\kappa_1]} (C_2 \times_{[\kappa_2]} C_3)$, $L_{12}$ be the behavior of $(C_1 \times_{[\kappa_1]} C_2)$ and $L_{13}$ be the behavior of $(C_1 \times_{[\kappa_1]} C_3)$.
   Then, 
   \begin{align*}
       L = & \{\sigma_1 [\cup] (\sigma_2 [\cup] \sigma_3) \mid \sigma_1 \in L_1,\ \sigma_2 \in L_2,\ \sigma_3 \in L_3,\ (\sigma_1, \sigma_2[\cup]\sigma_3) \in [\kappa_1],\  (\sigma_2, \sigma_3) \in [\kappa_2] \} \\
         = & \{\sigma_1 [\cup] (\sigma_2 [\cup] \sigma_3) \mid \sigma_1 \in L_1,\ \sigma_2 \in L_2,\ \sigma_3 \in L_3,\ (\sigma_1, \sigma_2) \in [\kappa_1],\ (\sigma_1,\sigma_3) \in [\kappa_1],\  (\sigma_2, \sigma_3) \in [\kappa_2] \} \\
         = & \{\sigma[\cup] \tau \mid \sigma \in L_{12},\ \tau \in L_{13},\ (\sigma, \tau) \in [\kappa_2] \} \\
         = & L' 
   \end{align*}
\end{proof}

\begin{lemma}
    \label{lemma:extension}
    Let $C_1 = (E_1, L_1)$ and $C_2 = (E_2, L_2)$ be two components.
    Let $\kappa^\ssync_\sqcap$ be a composability relation on observables with $\sqcap \subseteq \Po(E_1) \times \Po(E_2)$.
    Then, for any $\sqcap'$ with $\sqcap' \cap (\Po(E_1) \times \Po(E_2)) = \emptyset$, then:
    \[
        C_1 \times_{[\kappa^\ssync_{\sqcap}]} C_2 =  C_1 \times_{[\kappa^\ssync_{\sqcap \cup \sqcap'}]} C_2
    \]
\end{lemma}
\begin{proof}
    For any observations $((O_1,t_1), (O_2,t_2)) \in \kappa^\ssync_\sqcap$, we have that $((O_1,t_1), (O_2,t_2)) \in \kappa^\ssync_{\sqcap\cup\sqcap'}$ since $(O_1, O_2) \in \sqcap$ implies that $(O_1, O_2) \in \sqcap \cup \sqcap'$. Conversely, if $(O_1, O_2) \in \sqcap \cup \sqcap'$ and $\sqcap' \cap \Po(E_1) \times \Po(E_2) = \emptyset$, then $(O_1, O_2) \in \sqcap$.
    Thus, for any $(\sigma_1,\sigma_2)$, $ (\sigma_1,\sigma_2)\in [\kappa^\ssync_\sqcap]$ if and only if $(\sigma_1,\sigma_2) \in [\kappa^\ssync_{\sqcap\cup\sqcap'}]$.
\end{proof}

\subsection{Properties of TESs}\label{subsection:prop}
We distinguish two kinds of properties of TESs: properties that we call \emph{trace properties}, and properties on sets of \tes s that we call \emph{behavior properties}, which correspond to hyper-properties in~\cite{CS10}.
The generality of our model permits to interchangeably construct a component from a property and extract a property from a component.
As illustrated in Example~\ref{ex:coordination}, when composed with a set of interacting components, a component property constrains the components to only expose desired behavior (i.e., behavior in the property). 
In Section~\ref{section:formal-example}, we provide more intuition for the practical relevance of these properties.
\begin{definition}
    \label{def:properties}
    A trace property $P$ over a set of events $E$ is a subset $P \subseteq \TES{E}{\Rp}$.    
    A component $C = (E,L)$ satisfies a property $P$, if $L \subseteq P$, which we denote as $C \models P$.
\end{definition}
\begin{example}
We distinguish the usual \emph{safety} and \emph{liveness} properties~\cite{AS85,CS10}, and recall that every trace property can be written as the intersection of a safety and a liveness property.
Let $X$ be an arbitrary set, and $P$ be a subset of $\N \rightarrow X$.
Intuitively, $P$ is safe if every \emph{bad} stream not in $P$ has a finite prefix every completion of which is bad, hence not in $P$.
A property $P$ is a liveness property if every finite sequence in $X^*$ can be completed to yield an infinite sequence in $P$. \\
For instance, the property of terminating behavior for a component with interface $E$ is a liveness property, defined as:
$$
P_{\mathit{finite}}(E) = \{ \sigma \in \TES{E}{\Rp} \mid \exists n \in \mathbb{N}. \forall i > n.\ \pr_1(\sigma)(i) = \emptyset \}
$$
$P_{\mathit{finite}}(E)$ says that, for every finite prefix of any stream in $\TES{E}{\Rp}$, there exists a completion of that prefix with an infinite sequence of silent observations $\emptyset$ in $P_{\mathit{finite}}(E)$. 
\hfill$\blacksquare$
\end{example}
\newcommand{\Orch}{\mathit{Orch}}
\newcommand{\Coord}{\mathit{Coord}}
\begin{example}\label{ex:comp-prop}
A trace property is similar to a component, since it describes a set of \tes s, except that it is \emph{a priori} not restricted to any interface. 
A trace property $P$ can then be turned into a component, by constructing the smallest interface $E_P$ such that, for all $\sigma \in P$, and $i \in \mathbb{N}$, $\pr_1(\sigma)(i) \subseteq E_P$.
The component $C_P = (E_P, P)$ is then the componentized-version of property $P$.
    \hfill$\blacksquare$
\end{example}
\begin{lemma}\label{lemma:satisfiability}
Given a property $P$ over $E$, its componentized-version $C_P$ (see Example~\ref{ex:comp-prop}) and a component $C = (E,L)$, then $C \models P$ if and only if $C \cap C_P = C$. 
\end{lemma}
\begin{proof}
    We recall the definition of the intersection in Example~\ref{ex:intersection}. 
    For any two components $C_1 = (E_1, L_1)$ and $C_2 = (E_2, L_2)$, the intersection $C_1 \cap C_2$ is the component $C_1 \times_{([\kappa_\sqcap^\ssync], [\cap])} C_2 = (E_1 \cup E_2,L)$
 where $\sqcap \subseteq E_1 \times E_2$ is such that $(O,O) \in \sqcap$ for all non-empty $O \subseteq E_1 \cup E_2$. 
    Given that $\cap$ satisfies the condition for using corollary~\ref{corollary:lifted-sync}, the product $\cap$ is idempotent. Let $C \cap C_P = (E, L')$. If $(\sigma, \tau) \in L'$ then $\sigma = \tau$.
    Thus, $L' \subseteq L \cap L_P$.

    Alternatively, let $\sigma \in L \cap L_P$.
    We observe that at any point $n \in \N$, we have $(\sigma(n),\sigma(n)) \in \kappa^\ssync_\sqcap(E,E)$, since, given $\sigma(n) = (O_n,t_n)$ and the assumption on $\sqcap$, we have $(O_n, O_n) \in \sqcap$ or $O_n = \emptyset$. Therefore, $(\sigma, \sigma) \in [\kappa^\ssync_\sqcap]$. 

    We conclude that $L \cap L_P = L'$.
\end{proof}

\begin{example}\label{ex:coordination}
    We use the term \emph{coordination property} to refer to a property used in order to coordinate behaviors.
    Given a set of $n$ components $C_i = (E_i, L_i)$, $i \in \{1, ..., n\}$, a coordination property $\Coord$ for the composed components is a property over events $E = E_1 \cup ... \cup E_n$, i.e., $\Coord \subseteq \TES{E}{\Rp}$.
    
    Consider the synchronous interaction of the $n$ components and let $C = C_1 \bowtie C_2 \bowtie ... \bowtie C_n$ be their synchronous product.
    Typically, a coordination property will not necessarily be satisfied by the composite component $C$, but some of the behavior of $C$ is contained in the coordination property. 
    The coordination problem is to find (e.g., synthesize) an orchestrator component $\Orch = (E_O,L_O)$ such that $C \bowtie \Orch \models \Coord$.
    The orchestrator restricts the component $C$ to exhibit only the subset of its behavior that satisfies the coordination property. 
    In other words, in their composition, $\Orch$ coordinates $C$ to satisfy $\Coord$.
    As shown in Example~\ref{ex:comp-prop}, since $\Coord$ ranges over the same set $E$ that is the interface of component $C_1 \bowtie C_2\bowtie ... \bowtie C_n$, a coordination property can be turned into an orchestrator by building its corresponding component. 
    The coordination problem can be made even more general by changing the composability relations or the composition functions used in the construction of $C$. 
    \hfill$\blacksquare$
\end{example}
\newcommand{\iinsert}{\mathit{insert}}
\newcommand{\shift}{\mathit{shift}}

Trace properties are not sufficient to fully capture the scope of interesting properties of components of cyber-physical systems. Some of their limitations are highlighted in Section~\ref{section:formal-example}. To address this issue, we introduce \emph{behavior properties}, which are strictly more expressive than trace properties, and give two illustrative examples. 
\begin{definition}
    A behavior property $\phi$ over a set of events $E$ is a hyper-property $\phi \subseteq \Po(\TES{E}{\Rp})$. 
    A component $C = (E,L)$ satisfies a hyper-property $\phi$ if $L \in \phi$, which we denote as $C \mmodels \phi$.
\end{definition}
\begin{example}\label{ex:hyperproperties}
    A component $C = (E,L)$ can be oblivious to time\red{, as in Example~\ref{ex:cyber}}. 
    Any sequence of time-stamps for an acceptable sequence of observables is acceptable in the behavior of such a component. This ``obliviousness to time" property is not a trace property, but a hyper-property, defined as:
    \[
        \phi_{\shift}(E) := \{ Q \subseteq \TES{E}{\Rp} \mid \forall \sigma \in Q.  \forall t \in OS(\Rp). \exists \tau \in Q. 
                                                                    \pr_1(\sigma) = \pr_1(\tau) \land \pr_2(\tau) = t\}
                                                                \]
        Intuitively, if $C \mmodels \phi_{\shift}(E)$, then $C$ is independent of time.
        \hfill$\blacksquare$
    \end{example}
    \begin{example}
        \label{ex:insert}
        We use $\phi_{\iinsert}(X, E)$ to denote the hyper-property that allows for arbitrary insertion of observations in $X \subseteq \Po(E)$ into every \tes\ at any point in time, i.e., the set defined as:
        \[
        \begin{array}{rlcccccc}
             \{ Q \subseteq \TES{E}{\Rp} \mid \forall \sigma \in Q.  \forall i \in \mathbb{N}. \exists \tau \in Q.
                \exists x \in X. (\forall&\hspace{-0.2em} j < i.     & \sigma(j) &=& \tau(j)) &\land \\
                                 (\exists&\hspace{-0.2em} t \in \Rp. & \tau(i)   &=& (x,t))   &\land \\
                                 (\forall&\hspace{-0.2em} j \geq i.      & \tau(j+1) &=& \sigma(j))& \}
        \end{array}
    \]
        Intuitively, elements of $\phi_{\iinsert}(X,E)$ are closed under insertion of an observation $x \in X $ at an arbitrary time. 
        \short{We show in Section~\ref{section:formal-example} how we use this closure property to design cyber-physical components.}
        \hfill$\blacksquare$
\end{example}

\newcommand{\sstop}{\mathit{stop}}
\section{Application}\label{section:formal-example}

This section is inspired by the work on soft-agents~\cite{TAY15,KLAT19},
and elaborates on the more intuitive version that we presented in Section~\ref{section:problem}.
We show in Sections~\ref{subsection:components} and ~\ref{subsection:interactions} some expressions that represent interactive cyber-physical systems, and in Section~\ref{subsection:properties} we formulate some trace and behavior properties of those systems.
Through these examples, we show how we use component-based descriptions to model a simple scenario of a robot roaming around in a field while taking energy from its battery.
We structurally separate the battery, the robot, and the field as independent components, and we explicitly model their interaction in a specific composed system.

\subsection{Description of components}\label{subsection:components}
We give, in order, a description for a robot, a battery, and a field component.
    Each component reflects a local and concise view of the physical and cyber aspects of the system. 
\paragraph{Robot} A robot component $R$ is a tuple $(E_R,L_R)$ with:
\[
\begin{array}{rll}
    E_R &= \{(\rread(&\hspace{-0.7em}\loc,R); l), (\rread(\bat,R); b), (\move(R);(d, \alpha)), (\charge(R); s) \mid \\
        &                                             &l \in [0,20]^2,\ b \in \Rp,\ 
    d \in \{\No, \Ea, \We, \So \},\ \alpha \in \Rp,\ s \in \{\ON,\OFF\}\}\\
        L_R &= \{\sigma \in &\hspace{-1.7em}\mathit{TES}(E_R) \mid \forall i \in \N. \exists e \in E_R.\ \pr_1(\sigma)(i) = \{e\}  \}
\end{array}
\]
where the read of the position, the read of the battery, the move, and the charge events contain respectively the position of the robot as a pair of coordinates $l\in[0,20]^2$ grid; the remaining battery power $b$ (in \si{Wh}); the move direction as pair of a cardinal direction $d$ and a positive number $\alpha$ for a demand of energy (in \si{\watt}); and the charge status $s$ as $\ON$ or $\OFF$. 
Note that the set of TESs $L_R$ allows for arbitrary increasing and non-Zeno sequences of timestamp.

\paragraph{Battery} A battery component $B(C)$ with capacity $C$ (in \si{Wh}) is a tuple $(E_B, L_B)$ with:
\begin{align*}
    E_B &= \{ (\rread(B); l), (\discharge(B); \eta_d), (\charge(B); \eta_c) \mid  0 \leq l \leq C,\ \eta_c, \eta_d: \Rp \rightarrow \Rp\}\\
    L_B &= \{ \sigma \in \TES{E_B}{\Rp} \mid \forall i \in \N. \exists e \in E_B.\ \pr_1(\sigma)(i) = \{e\} \land P_{B}(\sigma) \}
\end{align*}
where the read, the charge, and the discharge events respectively contain the current charge status $l$ (in \si{Wh}), the discharge rate $\eta_d$, and the charge rate $\eta_c$.
The discharge and charge rates are coefficients that depend on the internal constitution of the battery, e.g., its current and voltage, and influence how the battery supplies energy to its user. The integration of $\eta_d$ (or $\eta_c$) on a time interval gives the power delivered (or received) by the battery in \si{Wh}.
The predicate $P_{B}(\sigma)$ guarantees that every behavior $\sigma$ of the battery satisfies the physical constraints for its acceptability. 
    An example for the structural constraint $P_{B}$ is that every $\rread$ event instance returns the battery level as a function of the occurrence of prior $\discharge$ and $\charge$ events.
    We introduce the $lev$ function, that takes a sequence of observations $s \in (\Po(\E) \times\Rp)^*$ of size $i > 1$ and returns the cumulative energy spent: $lev(s) = lev(\langle s(0),s(1)\rangle) + lev(\langle s(1) ...s(i)\rangle)$, where
    \[
    lev(\langle (O_1, t_1), (O_2, t_2)\rangle) = 
    \begin{cases} 
        \int_{t_1}^{t_2}(\eta_d(t) +\eta_I(t))dt &\ if\ (\discharge(B);\eta_d) \in O_1 \\
        \int_{t_1}^{t_2}(\eta_I(t) - \eta_c(t)) dt&\ if\ (\charge(B);\eta_c) \in O_1 \\
        \int_{t_1}^{t_2}\eta_I(t) dt &\ otherwise
    \end{cases}
    \]
    and $\eta_I$ is an internal discharge rate.
    The constraint $P_B$ is defined such that all $(\rread(B);l)$ events in $E_B$ return the current battery level of the robot, in accordance with the $lev$ function, i.e., for all $\sigma \in L_B$:
    \begin{align*}
        \forall i \in \mathbb{N}.  &(\rread(B);l) \in \pr_1(\sigma)(i) \implies l = \min(C, \max(C - lev(\langle \sigma(0), ..., \sigma(i)\rangle), 0))
    \end{align*}
    where $C$ is the maximum capacity of the battery.
The property $P_B$ assumes that initially at $t=0$ the battery is at its maximum charge $C$, that the battery level decreases after each discharge event, increases after each charge event, proportionally to the discharge and the charge rates. Moreover, a discharge below $0$ is physically forbidden. 
    Observe that different alternatives for the predicate $P_{B}$ account for different models of batteries. 
    Moreover, our model allows for specifications where the discharge factor depends on external parameters (temperature, discharge level, etc), adding a non-linear aspect to the model.

\paragraph{Field} A field component $F(l_0)$ contains a single object that we identify as $I$ initially at location $l_0$, has a fixed size of $[0,20]^2$, and contains a charging station at location $(5;5)$.
A field component is a tuple  $(E_F, L_F)$ with:
\begin{align*}
    E_F &= \{ (\loc(I);p), (\move(I);(d,F))  \mid  p \in [0,20]^2,\ d \in \{\No, \So, \Ea, \We\},\ 
    F_t \in \Rp\}\\
    L_F &= \{ \sigma \in \TES{E_F}{\Rp} \mid  \forall i \in \N. \exists e \in E_F.\ \pr_1(\sigma)(i) = \{e\} \land P_F(\sigma) \}
\end{align*}
where the $\loc$ and the $\move$ events respectively contain the position of object $I$ and the pair of a direction $d$ of the move of object $I$ and a force $F_t$ of traction applied by object $I$. 
A field has an internal friction factor $\mu$ whose value depends on the position on the field. 
With a friction value of $0$, the object will have no traction and thus will stay put in place instead of moving on the field (e.g., failure to move on a layer of ice). With a friction of $1$, the move event will displace the object proportionally to the force of the move (e.g., a move on a layer of concrete). A friction factor between $0$ and $1$ captures other scenarios between those two extremes (e.g., a move on a layer of grass).
The predicate $P_F(\sigma)$ guarantees that every behavior $\sigma$ satisfies the physics of the field component.
$P_F$ models the case where the object $I$ is initially at position $l_0$ and every move event changes continuously the location of the object on the field according to the direction $d$, the force of traction $F_t$, and the friction $\mu$. A move event has no effect if it occurs while the position of $I$ is on the boundary of the field: this scenario simulates the case of a fenced field, where moving against the fence has the same observable as not moving. 

The internal constraints of the field are such that the $\move$ observation triggers an internal displacement of object $I$ proportional to the force that the object has applied.
We write $\Delta d(t,t_0, (x_0,y_0))$ to denote the displacement from a time $t_0$ where the object is at rest at position $(x_0,y_0)$, to a time $t$, defined as:
\begin{align}
    m\vec{a} &= \vec{F_t}(x_0,y_0)\notag\\
    ma &=  F_t(x_0,y_0)\notag\\
    v(t,t_0) &=  (\cfrac{F_t(x_0,y_0)}{m})(t - t_0) \notag\\
    \Delta d(t,t_0,(x_0,y_0)) &=  \cfrac{1}{2}(\cfrac{F_t(x_0,y_0)}{m})(t - t_0)^2 
\end{align}
where $||\vec{F_t}|| \leq \cfrac{1}{4}\ \mu(x_0,y_0) m g$, e.g., the traction force on a wheel (supporting one fourth of the weight of the object) is less than the maximal friction force, with $\mu$ the friction coefficient, $m$ the mass of the object; and $F_t$ is the constant traction force of the object.
Observe that we chose to make the friction coefficient dependent on the initial position $x_0$ of the object before the move. This choice reflects the simplifying assumption that the friction will not substantially change during the movement.
Alternatively, one can imagine a different structure for the field component to support variable friction during a move in $P_F$.

\newcommand{\dis}{\mathit{dis}}
    An example for the constraint $P_{F}$ reflects the constraint that for each sequence of observations, the output value of a $read$ event corresponds to the current position of the robot given its previous moves.
    We will use a function called $\dis$ to determine the cumulative displacement of the robot after a sequence of observations.
    Let $s \in (\Po(\E)\times \Rp)^*$ be a finite sequence of observations of size $i\geq 1$. The displacement of the object $I$, at position $(x_0,y_0)$, after a sequence of events $s$ is given by $\dis(\langle(O_0,t_0)\rangle, (x_0,y_0)) = (x_0,y_0)$ and $\dis(s,(x_0,y_0)) = \dis(\langle s(0),...,s(i-1)\rangle, (x', y'))$, where for $s(0) = (O_0,t_0)$ and $s(1) = (O_1, t_1)$: 
    \[
    (x',y') = 
    \begin{cases} 
        (x_0,y_0+\Delta d(t_1,t_0, (x_0,y_0)))   &\  if\ (\move(I); (\No,F_t)) \in O_0 \\
        (x_0,y_0-\Delta d(t_1, t_0,(x_0,y_0))) &\  if\ (\move(I); (\So, F_t)) \in O_0 \\
        (x_0+\Delta d(t_1, t_0,(x_0,y_0)),y_0)  &\  if\ (\move(I); (\Ea,F_t)) \in O_0 \\
        (x_0-\Delta d(t_1, t_0,(x_0,y_0)),y_0) &\  if\ (\move(I); (\We, F_t)) \in O_0  \\ 
        (x_0,y_0) &\ otherwise
    \end{cases}
\]
    with $\Delta d(t,t_0,(x_0,y_0))$ defined in Equation~(1).
    $P_F$ is defined to accept all \tes s such that every $\rread$ event returns the current position of the robot on the field, according to its displacement over time. 
    Given $\sigma \in \TES{E_F}{}$, $P_F(\sigma)$ is true if and only if
    \[
    \begin{array}{rl}
        \forall i \in \N.&  (\loc(I);p) \in \pr_1(\sigma)(i) \implies p = |\dis(\langle \sigma(0) .... \sigma(i)\rangle, l_0)|_{[-20,20]}
    \end{array}
\]
    with 
    $
    |(x,y)|_{[-20,20]} = (\min(\max(x,-20),20), \min(\max(y,-20),20))
    $, and $l_0$ the initial position of object $I$.
    $P_F$ models the case where the robot starts in position $l_0$ and every move event changes the location of the robot on the field. 

Robots $R_1$ and $R_2$ are two instances of the robot component, where all occurrences of $R$ have been renamed respectively to $R_1$ and $R_2$ (e.g., $(\rread(\loc,R),l)$ becomes $(\rread(\loc,R_1),l)$ for the robot instance $R_1$, etc.).
    Similarly, we consider $B_1$ and $B_2$ to be two instances of the battery component $B$, and $F_1((0;0))$ and $F_2((5;0))$ to be two instances of the field component $F$ parametrized by the initial location for the object $I$, where the objects in fields $F_1$ and $F_2$ are renamed to $I_1$ and $I_2$, and respectively initialized at position $(0;0)$ and $(5;0)$.

\subsection{Interaction}\label{subsection:interactions}
We detail three points of interactions on observables among a robot and its battery, a robot and a field on which it moves, and two instances of a field component.
    The composability relations that relate the events of a robot, a battery, and a field impose some necessary constraints for the physical consistency of the cyber-physical system. For instance, that the power requested by the robot must match the characteristic of the battery.

\paragraph{Robot-battery} 
Interactions between a robot component and its battery are such that, for instance, every occurrence of a move event at the robot component must be simultaneous with a discharge event of the battery, with the discharge factor proportional to the demand of energy from the robot.
Given a robot component $R$ and a battery component $B$, we define the symmetric relation $\sqcap_{RB}$ on the set $\Po(E_R \cup E_B)$ to be the smallest relation such that:
\[
\begin{array}{rcll}
    \{(\rread(\bat,R);b)\} &\sqcap_{RB}& \{(\rread(B);b)       \} &\text{ for all } 0 \leq b \leq C \\
    \hspace{-0.4em}\{(\move(R);(d,\alpha))      \} &\sqcap_{RB}& \{(\discharge(B);\eta_d)\} &\text{ for all }d \in \hspace{-0.2em}\{\No,\So,\We,\Ea\} \\
    \hspace{-0.2em}\{(\charge(R);\ON)  \} &\sqcap_{RB}& \{(\charge(B);\eta_c)   \} &
\end{array}
\]
with $\eta_d(t) > \alpha$ for all $t \in \Rp$, e.g, the power delivered by the battery during a discharge is greater than the power required by the move; and with $C$ the capacity of the battery.

\paragraph{Robot-field} Interactions between a robot component and a field component are such that, for instance, every move event of the robot component must be simultaneous with a move event of the object $I$ on the field, with a variable friction coefficient.
Given a robot component $R$ and a field component $F$, we define the symmetric relation $\sqcap_{RF}$ on the set $\Po(E_R \cup E_F)$ to be the smallest relation such that:
\[
\begin{array}{rcll}
    \{(\rread(\loc,R);l) \} &\sqcap_{FR}& \{(\loc(I);l)      \} &\text{ for all }l \in [0,20]^2 \\
    \{(\move(R);(d,\alpha))       \} &\sqcap_{FR}& \{(\move(I);(d,F_t))     \} &\text{ for all }d \in \{\No,\We,\Ea,\So\}, v \in \Rp\\
    \{(\charge(R); \ON)  \} &\sqcap_{FR}& \{(\loc(I);(5,5))  \} &
\end{array}
\]
with $F_t = \cfrac{\alpha}{R\omega}$ with $R$ the radius of the wheels of the robot and $\omega$ the speed of rotation of the wheels (assumed to be constant during the move).
Observe that a robot can charge only if it is located at the charging station.

\paragraph{Field-field}
We add also interaction constraints between two fields, such that no observation can gather two read events containing the same position value.
Thus, given two fields $F_1$ and $F_2$, let $\sqcap_{F_{12}}$ be the smallest symmetric mutual exclusion relation on the set $\Po(E_{F_1} \cup E_{F_2})$ such that:
\[
\{(\loc(I_1);l)\} \sqcap_{F_{12}} \{(\loc(I_2);l)\} \text{ for all }l \in [0,20]^2.
\]
Observe that we interpret $\sqcap_{F_{12}}$ as a mutual exclusion relation. 
At first sight, the field does not prevent the two robots to share the same location. It only removes the possibility to observe the two robots at the same position. If, however, the field's behavior is closed under insertion of simultaneous read observables from both robots, then the two propositions stated above coincide (see Example~\ref{ex:hyperproperties}).

\paragraph{Product} We use set union as a composition function on observables: given two observables $O_1$ and $O_2$, we define $O_1 \oplus O_2$ to be the observable $O_1 \cup O_2$.
We use the synchronous and mutual exclusion composability relations on \tes s introduced in Definition~\ref{def:sync} and Definition~\ref{def:excl}. 
We represent the cyber-physical system consisting of two robots $R_1$ and $R_2$ with two private batteries $B_1$ and $B_2$, and individual fields $F_1$ and $F_2$, as the expression:
\begin{equation}
System = (F_1 \times_{[\kappa^{\mathit{excl}}_{\sqcap_{F_{12}}}]} F_2) \times_{[\kappa^{\ssync}_{(\sqcap_{F_1R_1} \cup \sqcap_{F_2R_2})}]}  ((R_1 \times_{[\kappa^{\ssync}_{\sqcap_{R_1B_1}}]} B_1) \times_{[\kappa_\top]} (R_2 \times_{[\kappa^{\ssync}_{\sqcap_{R_2B_2}}]} B_2))
\label{eq:cps}
\end{equation}

Note that the previous expression describes the same component as: 
\begin{equation}
    System = ((F_1 \times_{[\kappa^\ssync_{\sqcap_{F_1R_1}}]} R_1) \times_{[\kappa^{\ssync}_{\sqcap_{R_1B_1}}]} B_1) \times_{[\kappa^{\mathit{excl}}_{\sqcap_{F_{12}}}]} ((F_2 \times_{[\kappa^\ssync_{\sqcap_{F_2R_2}}]} R_2) \times_{[\kappa^{\ssync}_{\sqcap_{R_2B_2}}]} B_2)
\end{equation}
\setcounter{equation}{0} 
\begin{proof}
    \begin{align}
      &(F_1 \times_{[\kappa^{\mathit{excl}}_{\sqcap_{F_{12}}}]} F_2) \times_{[\kappa^{\ssync}_{(\sqcap_{F_1R_1} \cup \sqcap_{F_2R_2})}]}  ((R_1 \times_{[\kappa^{\ssync}_{\sqcap_{R_1B_1}}]} B_1) \times_{[\kappa_\top]} (R_2 \times_{[\kappa^{\ssync}_{\sqcap_{R_2B_2}}]} B_2))\\
        =\ &(F_1 \times_{[\kappa^{\mathit{excl}}_{\sqcap_{F_{12}}}]} F_2) \times_{[\kappa^{\ssync}_{(\sqcap_{F_1R_1} \cup \sqcap_{F_2R_2})}]}  ((R_1 \times_{[\kappa^{\ssync}_{\sqcap_{R_1B_1}}]} B_1) \times_{[\kappa^\ssync_{\sqcap_{F_{11}}\cup \sqcap_{F_{22}}}]} (R_2 \times_{[\kappa^{\ssync}_{\sqcap_{R_2B_2}}]} B_2))\\
        = \ &((F_1 \times_{[\kappa^{\mathit{excl}}_{\sqcap_{F_{12}}}]} F_2) \times_{[\kappa^{\ssync}_{(\sqcap_{F_1R_1} \cup \sqcap_{F_2R_2})}]}  (R_1 \times_{[\kappa^{\ssync}_{\sqcap_{R_1B_1}}]} B_1)) \times_{[\kappa^\ssync_{\sqcap_{F_{11}} \cup \sqcap_{F_{22}}}]} ((F_1 \times_{[\kappa^{\mathit{excl}}_{\sqcap_{F_{12}}}]} F_2) \times_{[\kappa^{\ssync}_{(\sqcap_{F_1R_1} \cup \sqcap_{F_2R_2})}]}  (R_2 \times_{[\kappa^{\ssync}_{\sqcap_{R_2B_2}}]} B_2))\\
        = \ &((F_1 \times_{[\kappa^{\mathit{excl}}_{\sqcap_{F_{12}}}]} F_2) \times_{[\kappa^{\ssync}_{\sqcap_{F_1R_1}}]}  (R_1 \times_{[\kappa^{\ssync}_{\sqcap_{R_1B_1}}]} B_1)) \times_{[\kappa^\ssync_{\sqcap_{F_{11}} \cup \sqcap_{F_{22}}}]} ((F_1 \times_{[\kappa^{\mathit{excl}}_{\sqcap_{F_{12}}}]} F_2) \times_{[\kappa^{\ssync}_{\sqcap_{F_2R_2}}]}  (R_2 \times_{[\kappa^{\ssync}_{\sqcap_{R_2B_2}}]} B_2))\\
        = \ &((R_1 \times_{[\kappa^{\ssync}_{\sqcap_{R_1B_1}}]} B_1) \times_{[\kappa^{\ssync}_{\sqcap_{F_1R_1}}]} (F_1 \times_{[\kappa^{\mathit{excl}}_{\sqcap_{F_{12}}}]} F_2)) \times_{[\kappa^\ssync_{\sqcap_{F_{11}} \cup \sqcap_{F_{22}}}]} ((F_1 \times_{[\kappa^{\mathit{excl}}_{\sqcap_{F_{12}}}]} F_2) \times_{[\kappa^{\ssync}_{\sqcap_{F_2R_2}}]}  (R_2 \times_{[\kappa^{\ssync}_{\sqcap_{R_2B_2}}]} B_2))\\
        = \ &(R_1 \times_{[\kappa^{\ssync}_{\sqcap_{R_1B_1}}]} B_1) \times_{[\kappa^{\ssync}_{\sqcap_{F_1R_1}}]} (F_1 \times_{[\kappa^{\mathit{excl}}_{\sqcap_{F_{12}}}]} F_2) \times_{[\kappa^{\ssync}_{\sqcap_{F_2R_2}}]}  (R_2 \times_{[\kappa^{\ssync}_{\sqcap_{R_2B_2}}]} B_2)\\
  = \ &((F_1 \times_{[\kappa^\ssync_{\sqcap_{F_1R_1}}]} R_1) \times_{[\kappa^{\ssync}_{\sqcap_{R_1B_1}}]} B_1) \times_{[\kappa^{\mathit{excl}}_{\sqcap_{F_{12}}}]} ((F_2 \times_{[\kappa^\ssync_{\sqcap_{F_2R_2}}]} R_2) \times_{[\kappa^{\ssync}_{\sqcap_{R_2B_2}}]} B_2)
    \end{align}
    where:
    \begin{itemize}
        \item $(1)$ to $(2)$ is given by the fact that $((R_1 \times_{[\kappa^{\ssync}_{\sqcap_{R_1B_1}}]} B_1) \times_{[\kappa_\top]} (R_2 \times_{[\kappa^{\ssync}_{\sqcap_{R_2B_2}}]} B_2)) = ((R_1 \times_{[\kappa^{\ssync}_{\sqcap_{R_1B_1}}]} B_1) \times_{[\kappa^\ssync_{\sqcap_{F_{11}}\cup \sqcap_{F_{22}}}]} (R_2 \times_{[\kappa^{\ssync}_{\sqcap_{R_2B_2}}]} B_2))$, i.e., synchronization on events that are not in the interface is the same as the relation $\kappa_\top$ (Lemma~\ref{lemma:extension});
        \item $(2)$ to $(3)$ is given by Lemma~\ref{lemma:distributivity};
        \item $(3)$ to $(4)$ simplifies in both side the synchronous composability relation to range over the interface of its operand components; 
        \item $(4)$ to $(5)$ commutativity of  $\times_{[\kappa^{\ssync}_{\sqcap_{F_1R_1}}]}$;
        \item $(5)$ to $(6)$ first rewrite $\kappa^{\ssync}_{\sqcap_{F_1R_1}}$, $\kappa^{\ssync}_{\sqcap_{F_2R_2}}$, and $\kappa^\ssync_{\sqcap_{F_{11}} \cup \sqcap_{F_{22}}}$ to all the same $\kappa^\ssync_{\sqcap}$ where $\sqcap = \sqcap_{F_{11}} \cup \sqcap_{F_{22}} \cup \sqcap_{F_1R_1} \cup \sqcap_{F_2R_2}$ which does not change the synchronization product (Lemma~\ref{lemma:extension}); and then uses associativity and idempotency of the product $\times_{[\kappa^\ssync_{\sqcap}]}$;
        \item $(6)$ to $(7)$ the synchronous products are distributed to the component on which they have effect only, i.e., $F_1$ and $F_2$ for $R_1$ and $R_2$ respectively. 
    \end{itemize}
\end{proof}

\newcommand{\Pswap}{\mathrm{P_{swap}}}
\newcommand{\id}{\mathrm{id}}
\newcommand{\System}{\mathit{System}}
\newcommand{\swap}{\mathit{swap}}
\subsection{Behavioral properties of components}\label{subsection:properties}
Let $E = E_{R_1} \cup E_{R_2} \cup E_{B_1} \cup E_{B_2} \cup E_{F_1} \cup E_{F_2}$ be the set of events for the composite system in Equation~(\ref{eq:cps}).
We formulate the scenarios described in Section~\ref{section:problem} in terms of a satisfaction problem involving a safety property on \tes s and a behavior property on the composite system.
We first consider two safety properties:
\[
P_{\energy} = \{ \sigma \in \TES{E}{\Rp} \mid \forall i \in \mathbb{N}. \{(\rread(B_1),0), (\rread(B_2),0)\} \cap \pr_1(\sigma)(i) = \emptyset \}
\]\[
P_{\nooverlap} = \{ \sigma \in \TES{E}{\Rp} \mid \forall i \in \mathbb{N}. \forall l \in [0,20]^2, \{ (\loc(I_1),l), (\loc(I_2),l) \} \nsubseteq  \pr_1(\sigma)(i) \} 
\]

The property $P_{\energy}$ collects all behaviors that never observe a battery value of $0\si{Wh}$.
The property $P_{\nooverlap}$ describes all behaviors where the two robots are never observed together at the same location.
Observe that, while both $P_{\energy}$ and $P_{\nooverlap}$ specify some safety properties, they are not sufficient to ensure the safety of the system.
We illustrate some scenarios with the property $P_{\energy}$.
If a component never reads its battery level, then the property $P_{\energy}$ is trivially satisfied, although effectively the battery may run out of energy. 
Also, if a component reads its battery level periodically, each of its readings may return an observation agreeing with the property. However, in between two read events, the battery may run out of energy (and somehow recharge).
To circumvent those unsafe scenarios, we add an additional behavioral property.

Let $X_{\rread} = \{ (\rread(B_1);l_1), (\rread(B_2);l_2) \mid 0 \leq l_1 \leq C_1,\ 0 \leq l_2 \leq C_2 \}$ be the set of reading events for battery components $B_1$ and $B_2$, with maximal charge $C_1$ and $C_2$ respectively. The property
$
\phi_{\mathit{insert}}(X_{\rread}, E)
$, as detailed in Example~\ref{ex:insert}, defines a class of component behaviors that are closed under insertion of $\rread$ events for the battery component.
Therefore, the system denoted as $C$, defined in Equation~\ref{eq:cps} is energy safe if $C \models P_{\energy}$ and its behavior is closed under insertion of battery read events, i.e., $C \mmodels \phi_{\mathit{insert}}(X_{\rread}, E)$.
In that case, every TES of the component's behavior is part of a set that is closed under insertion, which means all read events that the robot may do in between two events observe a battery level greater than $0\si{Wh}$.
The behavior property enforces the following safety principle: had there been a violating behavior (i.e., a run where the battery has no energy), then an underlying \tes\ would have observed it, and hence the behavioral property would have been violated.

Another scenario for the two robots is to consider their coordination in order to have them swap their positions. Let $F_1$ be initialized to have object $I_1$ at position $(0,0)$ and $F_2$ have $I_2$ at position $(5,0)$. The property of position swapping is a liveness property defined as:
\begin{align*}
    \Pswap = \{ \sigma \in \TES{E\cup\{\diamond\}}{\Rp} \mid &\{ (\loc(I_1),(0,0)), (\loc(I_2),(5,0))\} \subseteq \pr_1(\sigma)(0) \textit{ and } \\ 
    &\exists i \in \mathbb{N}. \{ (\loc(I_1),(5,0)), (\loc(I_2),(0,0)), \diamond \} \subseteq \pr_1(\sigma)(i) \}
\end{align*}
where $\diamond$ is used as an external symbol not in $E$.
It is sufficient for a liveness property to be satisfied for the system to be live, i.e., in the case of $\Pswap$, eventually reach a swapped position. However, it may be that the two robots swap their positions before the actual observation happens. In that case, using a similar behavioral property as for safety property will make sure that if there exists a behavior where robots swap their positions, then such behavior is observed  as soon as it happens.

Given a set of events $E$, let $\sqcap_E = \{ (O,O) \mid O \subseteq E\}$ be a relation on sets of events. 
Let $\Pswap\downarrow R_i\subseteq \TES{E_{R_i}\cup\{\diamond\}}{}$ be the projection of property $\Pswap$ on the set of events $E_{R_i}$ such that:
\[
    \tau \in \Pswap\downarrow R_i \iff \exists \sigma \in \Pswap. \forall n \in \N. (\sigma(n) = (O,t) \implies \tau(n) = (O\cap (E_{R_i}\cup\{\diamond\}),t))
\]
Let $C_\mathrm{swap}^{R_i} = (E_{R_i}\cup\{\diamond\}, \Pswap\downarrow R_i)$ and $C_\mathrm{swap} = (E_{R_1}\cup E_{R_2} \cup \{\diamond\}, \Pswap)$, 
then $C_\mathrm{swap} = (C_\mathrm{swap}^{R_i} \times_{\kappa^\ssync_{\sqcap_{\{\diamond\}}}} C_\mathrm{swap}^{R_i})$.
We show an equivalence between coordination of the two robots by a centralized coordinator (e.g., $C_\mathrm{swap}$), and coordination of the two robots by a decentralized coordination (e.g., $C_\mathrm{swap}^{R_1}$ and $C_\mathrm{swap}^{R_2}$).

Indeed, due to Lemma~\ref{lemma:distributivity}:
\begin{align}
    C_{\mathrm{swap}} \times_{[\kappa^\ssync_{\sqcap_{E_{R_1}}\cup \sqcap_{E_{R_2}}}]} (R_1 \times_{[\kappa^\ssync_\top]} R_2)
    = &(C^{R_1}_{\mathrm{swap}} \times_{[\kappa^\ssync_{\sqcap_{\{\diamond\}}}]} C^{R_2}_{\mathrm{swap}}) \times_{[\kappa^\ssync_{\sqcap_{E_{R_1}}\cup \sqcap_{E_{R_2}}}]} (R_1 \times_{[\kappa^\ssync_\top]} R_2)\\
    = &(C^{R_1}_{\mathrm{swap}} \times_{[\kappa^\ssync_{\sqcap_{E_{R_1}}}]} R_1) \times_{[\kappa^\ssync_{\sqcap_{\{\diamond\}}}]} (C^{R_2}_{\mathrm{swap}} \times_{[\kappa^\ssync_{\sqcap_{E_{R_2}}}]} R_2)
\end{align}

The above example shows how a property can be decomposed into sub-properties that interact via some shared events.
Such decomposition makes, a structural distinction between a \emph{local} form of coordination imposed by the products of the sub-property and its interacting component, and a \emph{global} form of coordination that coordinate the two properties using an additional $\diamond$ signal. 
Moreover, this scheme demonstrates the ability to localize a coordinator next to the components that it coordinates, which makes our design modular.

\section{Related and future work}
\label{section:related-work}

Our work offers a component-based semantics for cyber-physical systems~\cite{L08,KK12}. 
In~\cite{AH01}, a similar aim is pursued by defining an algebra of components using interface theory. Our component-based approach is inspired by~\cite{AR03,A05}, where a component exhibits its behavior as a set of infinite timed-data streams.
More details about co-algebraic techniques to prove component equivalences can be found in~\cite{R00}. 

Our model of components assumes some underlying physical
models. We do not give a precise account on how to model
physics of a component, in contrast to work using hybrid
models~\cite{H96}. TESs contrast to discrete event modes, such as
the work in~\cite{A02,L19,SSLST95}, by being based on arbitrary time
sampling strategies rather than event driven observations.
We abstract and generalize such work by supporting time
sensitive system evolution and more generally observation
of physical properties.

In~\cite{FLDL18}, the authors describe an algebra of timed machines and of networks of timed machines. 
A timed machine is a state based description of a set of timed traces, such that every observation has a time stamp that is a multiple of a time step $\delta$.
\short{
The work provides some decidability results on feasibility and consistency of networks of timed machines. 
}
This work differs from our current development in several respects. 
We focus in this paper on different algebraic operations on sets of timed-traces (TESs), and abstract away any underlying operational model (e.g., timed-automata).
In~\cite{FLDL18}, the authors explain how algebraic operations on timed machines \emph{approximate} the intersection of sets of timed-traces.
In our case, interaction is not restricted to input/output composition, but depends on the choice of a composability constraint on \tes s and a composition function on observables.
The work in~\cite{FLDL18} defines an interesting class of components (closed under insertion of silent observation - $r$-closed) that deserves investigation.

Cyber-physical systems have also been studied from an actor-model perspective, where actors interact through events~\cite{T08}. Methods for achieving synchronous behavior using asynchronous means of interaction are presented in~\cite{SASMO10}.

In \cite{kanovich-etal-16formats} a multiset rewriting model of time
sensitive distributed systems such as cyber-physical agent, is
introduced. Two verification problems are defined relative to a given
property P: realizability (is there a trace that satisfies P), and
survivability (do all traces satisfy P) and their complexity is analyzed.
In \cite{kanovich2021complexity} the theory is extended with two further
properties that concern the ability to avoid reaching a bad state.

Recent work has shown plenty of interest in studying the satisfaction problem of hyper-properties and the synthesis of reactive systems~\cite{BCJL20}. Some works focus more particularly on using hyper-properties for cyber-physical design~\cite{LJXJT17}.

The extension of hybrid automata~\cite{LSV03} to quantized hybrid automata is presented in~\cite{JAP19}, where the authors apply their model to give a formal semantics for data flow models of cyber-physical systems such as Simulink~\cite{DBLP:journals/sttt/HamonR07}.

Compared to formalisms that model cyber-physical systems as more concrete operational or state-based mechanisms, such as automata or abstract machines, our more general abstract formalism is based only on the observable behavior of cyber-physical components and their composition into systems, regardless of what more concrete models or mechanisms may produce such behavior.

For future work, we want to provide a finite description for components, and use our current formalism as its formal semantics. 
In fact, we first started to model interactive cyber-physical systems as a set of finite state automata in composition, but the underlying complexity of automata interaction led us to introduce a more abstract component model to clarify the semantics of those interactions.
Moreover, we want to investigate several proof techniques to show equivalences of components. We expect to be able to reason about local and global coordination, by studying how coordinators distribute over our different composition operators.
Finally, our current work serves as a basis for defining a compositional semantics for a state-based component framework~\cite{cp-agents} written in Maude~\cite{clavel-etal-07maudebook}, a specification and programming language based on rewriting logic. We plan to focus on evaluating the robustness of a set of components with respect to system requirements expressed as trace- or hyper-properties. The complexity of the satisfaction problem requires some run-time techniques to detect deviations and produce meaningful diagnoses~\cite{KLAT19}, a topic that we are currently exploring. 

\section{Conclusion}\label{section:conclusion}
This paper contains three main contributions. First, we introduce a component model for cyber-physical systems where cyber and physical processes are uniformly described in terms of sequences of observations.
Second, we provide ways to express interaction among components using algebraic operations, such as a parametric product and division, and
give conditions under which product is associative,
commutative, or idempotent.
Third, we provide a formal basis to study trace- and hyper-properties of components, and demonstrate the application of our work in an example describing several coordination problems. 

Our semantic model provides a formal basis for designing interacting cyber-physical systems, where interaction is defined explicitly and exogenously as an algebraic operation acting on components.
As a future step, we plan to use the semantic model introduced in this work to give a compositional semantics for interacting (state-based) specification for cyber-physical components. 
We aim to use our modular design in order to study problems of diagnosis in systems of interacting cyber-physical components.

\paragraph{Acknowledgement.}
Talcott was partially supported by the U. S. Office of Naval Research under award numbers N00014-15-1-2202 and N00014-20-1-2644, and NRL grant N0017317-1-G002.
Arbab was partially supported by the U. S. Office of Naval Research under award number N00014-20-1-2644.

\nocite{*}
\bibliographystyle{elsarticle-num}
\bibliography{references.bib}

\end{document}